\documentclass[letterpaper]{article} 
\usepackage{aaai24} 
\usepackage{times} 
\usepackage{helvet} 
\usepackage{courier} 
\usepackage[hyphens]{url} 
\usepackage{graphicx} 
\usepackage{natbib} 
\usepackage{caption} 
\usepackage{algorithm}
\usepackage{algorithmic}

\usepackage{newfloat}
\usepackage{listings}
\DeclareCaptionStyle{ruled}{labelfont=normalfont,labelsep=colon,strut=off} 
\floatstyle{ruled}
\newfloat{listing}{tb}{lst}{}
\floatname{listing}{Listing}
\usepackage{amsmath}
\usepackage{subfigure}
\usepackage{graphicx}
\usepackage{float}
\usepackage{subfig}
\usepackage[dvipnames,svgnames]{xcolor}
\usepackage{colortbl}
\usepackage{amsthm}
\usepackage{booktabs}
\usepackage{amssymb}
\usepackage{bbding}
\usepackage{dashrule}

\usepackage{amsmath,graphicx}
\usepackage{algorithm}
\usepackage{algorithmic}
\usepackage{amsthm}

\usepackage{subfigure}
\usepackage{amssymb}
\usepackage{xcolor}
\usepackage{booktabs}
\usepackage{multirow}

\usepackage[utf8]{inputenc} 
\usepackage[T1]{fontenc} 
\usepackage{url} 
\usepackage{amsfonts} 
\usepackage{nicefrac} 
\usepackage{microtype} 

\usepackage{diagbox}
\usepackage{tikz}
\usetikzlibrary{shadings}
\usepackage{amsmath}
\newtheorem{theorem}{Theorem}

\newcommand{\tabincell}[2]{
\begin{tabular}{@{}#1@{}}#2\end{tabular}
}
\title{Revisiting the Information Capacity of Neural Network Watermarks: \\ Upper Bound Estimation and Beyond}
\author{
Fangqi Li, Haodong Zhao, Wei Du, Shilin Wang\thanks{Shilin Wang is the corresponding author.} \\
}
\affiliations{
School of Electronic Information and Electrical Engineering, Shanghai Jiao Tong University \\
\{solour\_lfq, zhaohaodong, dddddw, wsl\}@sjtu.edu.cn
}

\begin{document}
%
\maketitle
\begin{abstract}
To trace the copyright of deep neural networks, an owner can embed its identity information into its model as a watermark.
The capacity of the watermark quantify the maximal volume of information that can be verified from the watermarked model.
Current studies on capacity focus on the ownership verification accuracy under ordinary removal attacks and fail to capture the relationship between robustness and fidelity.
This paper studies the capacity of deep neural network watermarks from an information theoretical perspective.
We propose a new definition of deep neural network watermark capacity analogous to channel capacity, analyze its properties, and design an algorithm that yields a tight estimation of its upper bound under adversarial overwriting.
We also propose a universal non-invasive method to secure the transmission of the identity message beyond capacity by multiple rounds of ownership verification. 
Our observations provide evidence for neural network owners and defenders that are curious about the tradeoff between the integrity of their ownership and the performance degradation of their products.
\end{abstract}


\section{Introduction}
\label{sec:1}
The intellectual property protection of artificial intelligence models, especially deep neural networks (DNN), is drawing increasing attention since the expense of building large models has become prohibitively high.
For example, the training of GPT-4 involves over 24,000 graphic processing units and more than 45TB manually proofread data~\cite{liu2023summary}.
If internal enemies steal and distribute the model, watermarking schemes continue to safeguard the intellectual property.

As demonstrated in Fig.~\ref{figure:1}, a DNN watermarking scheme adds the owner's identity message into the model to be protected.
Once the model is stolen, the owner can claim its copyright publicly by requesting a third-party judge to verify the identity message hidden in the victim model~\cite{main:7}.
Watermarking schemes have been designed for various kinds of DNN models including image classifier~\cite{main:8}, image generator~\cite{main:27}, pretrained natural language encoder~\cite{plmmark}, graph neural networks~\cite{main:29}, etc.

Most discussions on the applicability of DNN watermarking schemes focus on unambiguity~\cite{main:56} and robustness~\cite{main:5}.
It has been proven that incorporating pseudorandomness and connecting the triggers into a chain~\cite{main:6} result in security against ambiguity attacks, i.e., an unauthorized party cannot pretend itself as the owner of an arbitrary DNN model.
A robust ownership proof remains valid even if the watermarked DNN model undertakes adversarial modifications.
There have been many methods to foster robustness including adding regularizer~\cite{main:23}, training with surrogate models~\cite{main:32}, using error correction code for the identity message~\cite{main:44}, etc.

\begin{figure}[!t]
\centering
\includegraphics[width=7.2cm]{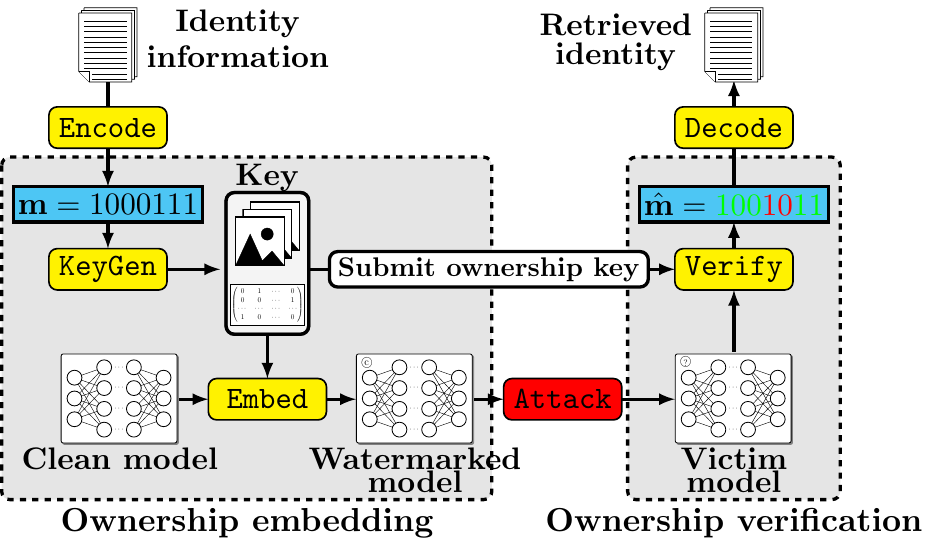}
\caption{The workflow of a DNN watermarking scheme.}
\label{figure:1}
\end{figure}

Nonetheless, the information capacity of DNN watermark, i.e., the maximal number of bits that can be accurately transmitted, has not undergone sufficient studies. 
So far, there is no uniform definition of capacity, especially for schemes that rely on backdoors.
As a result, a fair comparison between different schemes regarding capacity is intractable.
Moreover, the overlook of information theory-based capacity makes it hard to determine the configurations of DNN watermarking schemes with regard to the identity information and the expense of copyright protection. 

To tangle with these deficiencies, we revisit the capacity of the DNN watermark.
The contributions of this paper are:
\begin{itemize}
\item We give the first information theory-based definition of DNN watermark capacity and demonstrate how the expense of copyright protection is measured in the degradation of the watermarked model's performance.
\item We design a capacity estimation algorithm that yields a tight upper bound of the capacity of DNN watermarks under adversarial overwriting.
\item We propose a variational approximation-based method to increase the accuracy of identity message transmission beyond the capacity by multiple rounds of ownership verification. 
It can be generalized to arbitrary DNN watermarking schemes in a non-invasive manner.
\end{itemize}

\section{Preliminaries}
\label{sec:2}
\subsection{DNN watermark}
A DNN watermarking scheme adds the owner's identity information, encoded as a binary string $\mathbf{m}$ with length $L$, to a clean DNN model $M_{\text{clean}}$ to produce a watermarked model $M_{\text{WM}}$.
The owner first generates an ownership key $K$ as the intermedium. 
\begin{equation}
\label{equation:1}
K\!\leftarrow\!\texttt{KeyGen}(L).
\end{equation}
The owner then tunes $M_{\text{clean}}$ into $M_{\text{WM}}$ with the original training loss $\mathcal{L}_{0}$ and an additional regularizer $\mathcal{L}_{\text{WM}}$ as watermark embedding.
This step minimizes:
\begin{equation}
\label{equation:2}
\mathcal{L}_{0}(M_{\text{WM}})\!+\!\lambda\!\cdot\!\mathcal{L}_{\text{WM}}(M_{\text{WM}},K,\mathbf{m}).
\end{equation}
Upon piracy, the owner submits $K$ and requests a judge to retrieve the ownership information from a suspicious victim model $M$ with a verifier  that usually takes the form:
\begin{equation}
\label{equation:3}
\texttt{Verify}(M,K)=\arg\min_{\mathbf{m}'}\left\{ \mathcal{L}_{\text{WM}}(M,K,\mathbf{m}') \right\}.
\end{equation}
It is expected that if $M$ is a copy or a slightly modified version of $M_{\text{WM}}$ then $\texttt{Verify}(M,K)\!=\!\mathbf{m}$.
A complete DNN watermarking scheme is featured by Eq.~\eqref{equation:1}\eqref{equation:2}\eqref{equation:3}.

In scenarios where the judge has access to parameters in the victim model, the ownership key is usually defined as a pseudorandomly generated matrix and a bias term $K\!=\!(\mathbf{X},\mathbf{b})$.
The watermark embedding regularizer for white-box DNN watermarking schemes subject to this assumption is
\begin{equation}
\label{equation:4}
\mathcal{L}_{\text{WM}}(M_{\text{WM}},K,\mathbf{m})\!=\!f(\sigma(\mathbf{X}\!\cdot\!\mathbf{W}\!+\!\mathbf{b}),\mathbf{m}),
\end{equation}
where $\mathbf{W}$ denotes certain parameters in $M_{\text{WM}}$~\cite{main:38} or the outputs of $M_{\text{WM}}$'s intermediate neurons~\cite{main:46}.
Considering $\mathbf{m}$ as a list of binary labels, $f$ can be cross-entropy loss~\cite{main:38}, hingle loss~\cite{main:56}, or a neural network classification backend~\cite{main:45}.

If the judge has only black-box access to the victim model, the ownership key is a set of triggers $K\!=\!\left\{\mathbf{t}_{n} \right\}$. 
The embedding loss takes the form:
\begin{equation}
\label{equation:5}
\mathcal{L}_{\text{WM}}\left(M_{\text{WM}},K,\mathbf{m}\right)\!=\!f\left(g\left(\left\{M_{\text{WM}}(\mathbf{t}_{n}) \right\}\right),\mathbf{m}\right),
\end{equation}
in which $g$ interprets the outputs of the triggers into a list of bits.
For example, for an eight-class classification model, a vanilla interpreter maps the prediction of each trigger into three bits. 
So $\mathbf{m}$ is compactly represented by the labels of $\nicefrac{L}{3}$ triggers and watermark embedding is equivalent to fine-tuning to fit the triggers.
In cases of image generators, the interpreters are complex decoders as in steganography~\cite{zj}.

In both white-box and black-box cases, it is necessary that the gradient of the model's parameters w.r.t. Eq.~\eqref{equation:4}\eqref{equation:5} is computable so the watermark embedding process becomes an optimization task.

\subsection{Information capacity of watermark}
Capacity is a fundamental property of watermark along with fidelity and robustness~\cite{lei2019multipurpose}.
It regulates the upstream identity information encoder since a message carrying more information than the capacity cannot be transmitted losslessly.
As an example, in the scope of digital image watermarking, the capacity of a pixel is subject to both the pixel's contribution to the overall visual quality and its stability under standard image processing schemes such as compression or obfuscation~\cite{moulin2001role}.
Embedding too many bits into one pixel results in visible distortion but embedding too few bits might fail to deliver the message.


\subsection{Related works}

Li et al. derived an upper bound of DNN watermark capacity by examining if parameters where the watermarks are embedded have collusions or not~\cite{main:68}.
Many established works treat the length of the secret message as an upper bound of the capacity~\cite{main:42}.
In black-box schemes, the capacity is straightforward measured by the number of triggers~\cite{main:19}.
Unlike in traditional watermark, these definitions overlook the correlation between the damage of watermarking to the performance of the model and the influence of adversarial modifications.

\section{Information Capacity of DNN Watermark}
\label{sec:3}
\subsection{Definition}
We consider DNN watermark as a channel from the identity message $\mathbf{m}$ to the judge's observation $\hat{\mathbf{m}}$, both with length $L$.
The adversarial modication $M_{\text{WM}}\!\rightarrow\! M_{\text{WM}}\!+\!\theta$ is featured by the parameter deviation $\theta$. 
Fixing the upper bound of the victim model's performance degradation by $\delta\!\geq\! 0$, the capacity of the watermark is the maximal volume of information that can be correctly transmitted through the channel when the model's performance declines by no more than $\delta$ after undertaking arbitrary attacks. 
This property can be characterized through the channel capacity:
\begin{equation}
\label{equation:6}
C(\delta,L)=\min_{\theta}\left\{\max_{p(\mathbf{m})} I(\mathbf{m};\hat{\mathbf{m}}) \right\},
\end{equation}
subject to:
\begin{equation}
\nonumber
\begin{aligned}
&\hat{\mathbf{m}}\!=\!\texttt{Verify}(M_{\text{WM}}\!+\!\theta,K),\\
&E(M_{\text{WM}}\!+\!\theta)\!\geq\! E(M_{\text{WM}})\!-\!\delta,
\end{aligned}
\end{equation}
in which $E(\cdot)$ is the performance evaluation metric of the DNN model, $p(\mathbf{m})$ is the distribution over $\mathbf{m}$, and $I(\cdot;\cdot)$ denotes mutual information.
The capacity of the DNN watermark satisfies the following properties.

\begin{theorem}
\label{theorem:1}
\textbf{(Monotonicity)}
$0\!\leq\! C(\delta,L)\!\leq\! L$.
$C(\delta,L)$ decreases in $\delta$.
$C(\delta,L)$ increases in $L$ if each bit of the identity message is independently embedded and retrieved.
\end{theorem}

\begin{proof}
Denote $\Theta(\delta)\!=\!\left\{\theta\!:\! E(M_{\text{WM}}\!+\!\theta)\!\geq\! E(M_{\text{WM}})\!-\!\delta\right\}$ and $u(\theta,L)=\max_{p(\mathbf{m})}I(\mathbf{m};\hat{\mathbf{m}})$. 

As the mutual information between $L$ binary random variables, $0\!\leq\!I(\mathbf{m};\hat{\mathbf{m}})\!\leq\!L$, so does $C(\delta,L)$.

For the second statement, we prove that if $0\!\leq\!\delta_{1}\!\leq\!\delta_{2}$ then $C(\delta_{1},L)\!\geq\!C(\delta_{2},L)$.
Since $\delta_{1}\!\leq\! \delta_{2}$ implies $\Theta(\delta_{1})\!\subset\!\Theta(\delta_{2})$, it is evident that: 
\begin{equation}
\nonumber
C(\delta_{1},L)\!=\!\min_{\theta\in\Theta(\delta_{1})}\!\left\{u(\theta\!,L) \right\}\!\geq\! \min_{\theta\in\Theta(\delta_{2})}\!\left\{u(\theta\!,L) \right\}\!=\!C(\delta_{2},L).
\end{equation}

Finally, we prove $C(\delta,L\!+\!1)\!\geq\! C(\delta,L)$.
By definition, there is $\theta'$ such that the capacity is $C(\delta,L\!+\!1)$ when the length of the identity message is $L\!+\!1$ and $E(M_{\text{WM}}\!+\!\theta')\geq E(M_{\text{WM}})\!-\!\delta$.
The DNN model watermarked with a message of length $(L\!+\!1)$ can be viewed as a model watermarked with a message of length $L$ by considering only the first $L$ bits during ownership verification, i.e.,
\begin{equation}
\nonumber
C(\delta\!,\!L)=\min_{\theta\in\Theta(\delta)}\!\left\{u(\theta,\!L) \right\}\!\leq \!u(\theta'\!,\!L)\!\leq\! u(\theta'\!,\!L\!+\!1)\!=\!C(\delta,\!L\!+\!1),
\end{equation}
where the second inequality holds since the $(L\!+\!1)$-th independent bit brings additional mutual information.
\end{proof}

\begin{theorem}
\label{theorem:2}
\textbf{(BER upper bound)}
Denote the bit error rate corresponding to adversarial modification $\theta$ as:
\begin{equation}
\label{equation:7}
\epsilon(\theta)=\min\left\{\frac{1}{2},\frac{\|\mathbf{m}\oplus\texttt{Verify}(M_{\text{WM}}\!+\!\theta,\!K) \|_{0}}{L}\right\}.
\end{equation}
Denote $\epsilon_{\delta}=\max_{\theta\in\Theta(\delta)}\left\{\epsilon(\theta) \right\}$ as the maximal BER when the model's performance drops for no more than $\delta$.
The capacity of the DNN watermark is upper bounded by:
\begin{equation}
\label{equation:8}
C(\delta,L)\!\leq\! L\!\cdot\!(1\!-\!H(\epsilon_{\delta})),
\end{equation}
with $H(x)\!=\!-x\!\cdot\!\log_{2}\!x\!-\!(1\!-\!x)\!\cdot\!\log_{2}(1\!-\!x)$.
If each bit of the identity message is independently embedded and verified then Eq.~\eqref{equation:8} is an equality.
\end{theorem}

\begin{proof}
DNN watermark can be viewed as the combination of $L$ binary symmetric channels, where each channel's average capacity is $(1\!-\!H(\epsilon_{\delta}))$ since the adversary that minimizes the mutual information exerts the maximal BER. 
The capacity of $L$ paralleling channels is no larger than the summation of their capacities, this yields Eq.~\eqref{equation:8}.
If each bit of $\mathbf{m}$ is independent of each other then the capacity equals the summation of all independent channels.
\end{proof}

Assume the owner's identity information contains $J$ bits, whose source might be a digital copyright administrator authority as shown in Fig.~\ref{figure:1}.
To resist adversarial removals, the owner encodes its identity into $\mathbf{m}$ with length $L\!\geq\! J$ by injecting redundancy.
The cost in fidelity is measured by the performance degradation due to watermark embedding $F(L)\!=\!E(M_{\text{clean}})\!-\!E(M_{\text{WM}})$.
Intuitively, $F(L)$ increases in $L$.
On the other hand, the scheme's robustness is reflected by the performance degradation that the adversary has to undertake to sabotage the ownership, i.e., $\min_{\delta}\left\{\delta\!:\!C(\delta,L)\!\leq\! J \right\}$, which increases in $L$ due to Theorem~\ref{theorem:1}.

The owner's objective is to protect the intellectual property of any DNN model whose performance is comparable to its state-of-the-art model $M_{\text{clean}}$.
Models whose functionality has been severely damaged are not worth protecting.
Formally, the tradeoff between fidelity and robustness in the context of capacity is described in the following theorem.

\begin{figure}[!t]
\centering
\includegraphics[width=8cm]{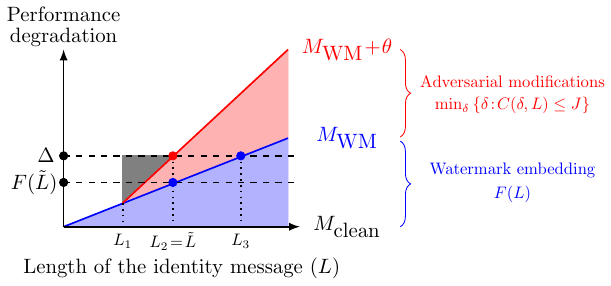}
\caption{The performance degradation v.s. the length of the identity message.
The blue area and the red area represent the cost in fidelity and robustness respectively.
}
\label{figure:2}
\end{figure}

\begin{theorem}
\label{theorem:3}
\textbf{(The minimal length of the identity message)}
Assume the owner's identity information contains $J$ bits and the owner wants to protect the copyright of all variants of $M_{\text{WM}}$ with performance degradation no more than $\Delta$ compared to $M_{\text{clean}}$. 
The necessarily minimal length of the identity message is:
\begin{equation}
\label{equation:9}
\tilde{L}\!=\!\min_{L}\left\{L\!:\!\left(F(L)\!+\!\min_{\delta}\left\{\delta\!:\!C(\delta,L)\!\leq\! J \right\}\right)\!\geq\! \Delta \right\}, 
\end{equation}
and the necessarily minimal expense of copyright protection measured in the model's performance degradation is $F(\tilde{L})$. 
\end{theorem}

\begin{proof}
The monotonicity of fidelity and robustness are visualized in Fig.~\ref{figure:2}, where $L_{1}\!=\!J$, $L_{2}\!=\!\tilde{L}$, and $L_{3}\!=\!\min_{L}\!\left\{L\!:\!F(L)\!\geq\!\Delta \right\}$.

If $L\!\in\![0,L_{1})$ then the identity information cannot be losslessly encoded as $\mathbf{m}$.
If $L\!\in\![L_{1},L_{2})$ then the adversary can find a model whose performance degradation is no more than $\Delta$ yet escapes the ownership verification.
Such cases constitute the gray area in Fig.~\ref{figure:2}.
If $L\!\in\![L_{2},L_{3})$ then the adversary cannot reduce the capacity below $J$ to escape the ownership verification unless sacrificing the model's performance for more than $\Delta$.
Finally, $L\!\in\![L_{3},\infty)$ is unacceptable since the damage caused by watermarking is too much.

Therefore, $\tilde{L}$ is minimal length of the identity message that satisfies the owner's purposes with the smallest performance degradation and sufficient robustness. 
\end{proof}


\subsection{Capacity estimation}
The capacity defined by Eq.~\eqref{equation:6} can hardly be analytically computed. 
Instead, we resort to Theorem~\ref{theorem:2} by measuring the correlation of the BER of the identity message with the performance degradation and conducting interpolation afterwards. 
An adversarial modification strategy $\mathcal{A}$ is evoked to drive the variations.
This estimation is formulated in Algo.~\ref{algorithm:1}.

\begin{theorem}
\label{theorem:4}
The estimated capacity $\hat{C}(\delta,L)$ by Algo.~\ref{algorithm:1} is an upper bound of $C(\delta,L)$.
\end{theorem}

\begin{proof}
The adversarial modification strategy $\mathcal{A}$ exerts one certain family of parameter perturbation $\hat{\Theta}(\delta)\!\subset\!\Theta(\delta)$.
The BER exclusively depends on $\theta$ according to Eq.~\eqref{equation:7}, so:
\begin{equation}
\nonumber
\epsilon_{\delta}=\max_{\theta\in\Theta(\delta)}\left\{\epsilon(\theta) \right\}\geq\max_{\theta\in\hat{\Theta}(\delta)}\left\{\epsilon(\theta) \right\}=\hat{\epsilon}_{\delta},
\end{equation}
since the r.h.s. of Eq.~\eqref{equation:8} is a monotonically decreasing function in $\epsilon$ (when $\epsilon<=0.5$), we have:
\begin{equation}
\nonumber
C(\delta,L)\!\leq\! L\cdot(1\!-\!H(\epsilon_{\delta}))\!\leq\! L\cdot(1\!-\!H(\hat{\epsilon}_{\delta}))\!=\!\hat{C}(\delta,L).
\end{equation}
\end{proof}

\begin{algorithm}[!t]
\caption{Capacity estimation.}
\label{algorithm:1}
\begin{algorithmic}[1]
\REQUIRE $M_{\text{WM}}$, $K$, $\mathbf{m}$, $L$, $\mathcal{A}$
\ENSURE Capacity estimation $\hat{C}(\delta,L)$
\STATE Initialize a memory $\mathbf{D}=\left\{ \right\}$.
\STATE Initialize $\hat{M}=M_{\text{WM}}$, $\delta=0$.
\STATE $\hat{\epsilon}_{0}=\min\left\{\frac{1}{2},\frac{\|\mathbf{m}\oplus \texttt{Verify}(\hat{M},K) \|_{0}}{L} \right\}$.
\STATE Save $\langle 0,\hat{C}(0,L)\!=\!L\!\cdot\!(1\!-\!H(\hat{\epsilon}_{0})) \rangle$ in $\mathbf{D}$.
\WHILE {$\hat{\epsilon}_{\delta}<= 0.5$}
\STATE Conduct adversarial modification $\hat{M}\leftarrow \mathcal{A}(\hat{M})$.
\STATE $\delta\!=\!E(M_{\text{WM}})\!-\!E(\hat{M})$.
\STATE Compute $\hat{\epsilon}_{\delta}$ as line 3.
\STATE Save $\langle \delta,\hat{C}(\delta,L)\!=\!L\!\cdot\!(1\!-\!H(\hat{\epsilon}_{\delta})) \rangle$ in $\mathbf{D}$.
\ENDWHILE
\STATE Return $\mathbf{D}$.
\end{algorithmic}
\end{algorithm}

If the adversarial modification strategy $\mathcal{A}$ is not destructive enough, i.e., $\hat{\epsilon}_{\delta}\ll \epsilon_{\delta}$, then the capacity would be overestimated.
Consequently, the minimal length of the identity message recommended by Theorem~\ref{theorem:3} after replacing $C(\delta,L)$ by $\hat{C}(\delta,L)$ is smaller and unsafe. 
To tighten the bound, we conduct an adversarial overwriting attack that directly embeds another randomly generated message into the pirated model with the same ownership key as shown in Table~\ref{table:attacks}.

In contrast to other advanced attacks such as reverse engineering~\cite{main:43}, or functionality equivalence attack~\cite{neuronmap}, adversarial overwriting can be directly generalized to any gradient-based DNN watermarking scheme by manipulating the watermark embedding process.
Adversarial overwriting adopts the strongest threat model where the adversary has full knowledge of the ownership key and covers the worst cases, e.g., the ownership verification has been exposed to malicious parties or there are internal enemies. 
Therefore, the bound is expected to be the tightest and provides a reliable reference for owners. 

\begin{table}[!t]
\centering
\caption{Comparison between adversarial modifications.
$\mathcal{L}_{0}$ is the model's original training loss function.
$\mathbf{m}_{\text{Adv}}$ and $K_{\text{Adv}}$ are the adversary's identity message and ownership key.
}
\scalebox{0.75}{
\begin{tabular}{c|c}
\toprule
\textbf{Attack method} & \textbf{Loss function to be optimized}\\
\toprule
Fine-tuning & $\mathcal{L}_{0}(M_\text{WM})$ \\
Overwriting & $\mathcal{L}_{0}(M_\text{WM})\!+\!\lambda\!\cdot\! \mathcal{L}_{\text{WM}}(M_\text{WM},K_{\text{Adv}},\mathbf{m}_{\text{Adv}})$ \\
\textbf{Adversarial overwriting} & $\mathcal{L}_{0}(M_\text{WM})\!+\!\lambda\!\cdot\! \mathcal{L}_{\text{WM}}(M_\text{WM},K,\mathbf{m}_{\text{Adv}})$ \\
\bottomrule
\end{tabular}}
\label{table:attacks}
\end{table}

\begin{figure}[!tbp]
\centering
\subfigure[]{
\begin{minipage}[htbp]{0.31\linewidth}
\centering
\includegraphics[width=2.3cm,height=2.4cm]{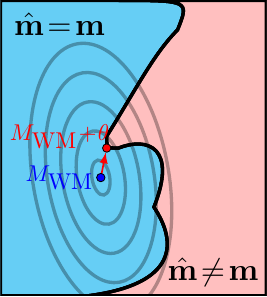}
\end{minipage}
}\subfigure[]{
\begin{minipage}[htbp]{0.31\linewidth}
\centering
\includegraphics[width=2.3cm,height=2.4cm]{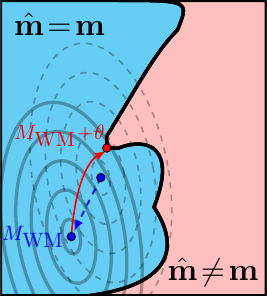}
\end{minipage}
}\subfigure[]{
\begin{minipage}[htbp]{0.31\linewidth}
\centering
\includegraphics[width=2.3cm,height=2.4cm]{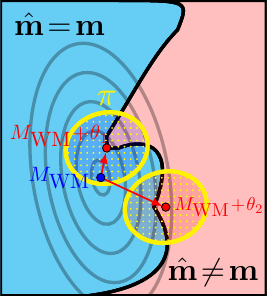}
\end{minipage}
}
\caption{
The contours denote levels of performance degradation.
(a) An adversarial modification.
(b) Increasing robustness as a defense.
(c) Averaging multiple rounds of ownership verification as a defense. $M_{\text{WM}}\!+\!\theta_{1}$ denotes a failed attack. 
$M_{\text{WM}}\!+\!\theta_{2}$ denotes a successful attack. 
}
\label{figure:3}
\end{figure}

\subsection{Breaking the capacity bottleneck}
Let the number of all legal owners or models be $2^{J}$, so the identity information of each owner contains $J$ bits. 
In one-time DNN ownership verification, the code rate, i.e., the volume of information transmitted in each round of communication, is $\frac{\log_{2} 2^J}{1}\!=\!J$ bits. 
According to basic information theory, the code rate cannot exceed the capacity for accurate communication, therefore $J\!\leq\! C(\delta,L)\!\leq\! L$.

When $L$ is fixed (so the price of watermarking is no more than $F(L)$), there are two approaches to increasing the number of entities that can use the copyright protection service. 
The first is to increase the capacity, whose necessary prerequisite is reducing the BER by Theorem~\ref{theorem:2}. 
Yet this additional robustness might increase the performance degradation due to watermark embedding as shown in Fig.~\ref{figure:3}(b). 

The other approach is to retrieve the identity message for $R$ times and average the results. 
The motivation is that the adversary erases the watermark with the smallest performance degradation, so the victim model falls at the decision boundary close to the watermarked model. 
So it is expected that the correct identity message can still be retrieved from neighbours of the victim model as shown in Fig.~\ref{figure:3}(c). 
Under this setting, the code rate becomes $\nicefrac{J}{R}$ bits and at most $2^{L}$ owners or models are verifiable with $R\!\geq\!\frac{L}{C(\delta,L)}$. 

However, the adversary is not obligated to use random modification during multiple rounds of ownership verification. 
So the error pattern $\mathbf{m}\oplus\hat{\mathbf{m}}$ could be fixed and cannot be eliminated by averaging.  
Instead, it is the judge's responsibility to incorporate randomness into the channel by perturbing the victim model and alleviating adversarial influences.

To randomize the noises in the channel, the judge incorporates parameter deviations $\mu$ following a distribution $\pi$.
In ownership verification, the identity message is retrieved by averaging the results as decoding the error correction code.
Concretely, the $i$-th bit is determined by the majority voting:
\begin{footnotesize}
\begin{equation}
\label{equation:10}
\hat{\mathbf{m}}[i]=\arg\max_{b\in\left\{0,1\right\}}\left\{\sum_{r=1}^{R}\mathbb{I}\left[\texttt{Verify}(M\!+\!\mu_{r},\!K)[i]=b\right]\right\}.
\end{equation}
\end{footnotesize}
The owner can further reduce the sensitivity of the identity message against potential adversarial modifications with the following regularizer during watermark embedding:
\begin{equation}
\label{equation:11}
\mathcal{L}_{\text{WM}}^{M}(M_{\text{WM}},K,\mathbf{m})\!=\!\frac{1}{P}\sum_{p=1}^{P}\mathcal{L}_{\text{WM}}(M_{\text{WM}}\!+\!\mu_{p},\!K,\!\mathbf{m}).
\end{equation}
It has be proven that if $\pi$ is a normal distribution $\mathcal{N}(0,\sigma^{2}\mathbf{I})$ and the scale of adversarial modification is restricted by $\|\theta\|_{2}\leq \rho$ then the BER of the estimation Eq.~\eqref{equation:10} is upper bounded by
$\inf\left\{y:\text{Pr}\left\{\epsilon(\theta)\geq y \right\}\leq\Phi\left(-\nicefrac{\rho}{\sigma}\right) \right\}$, where $\Phi(\cdot)$ is the culmulative distribution function of a normal distribution~\cite{main:17}. 


\begin{table}[!t]
\centering
\caption{Comparison between configurations of the watermark embedding regularizer and the verification formula. 
MROV and MROV-V denote \textbf{M}ultiple \textbf{R}ounds of \textbf{O}wnership \textbf{V}erification and its \textbf{V}ariational version respectively. 
}
\scalebox{0.75}{
\begin{tabular}{c|c|c|c}
\toprule
\textbf{Configuration} & \tabincell{c}{\textbf{Verification}\\ \textbf{formula}} & \tabincell{c}{\textbf{Embedding}\\ \textbf{regularizer}} & \tabincell{c}{\textbf{Applicability to}\\ \textbf{black-box}} \\
\toprule
Baseline & Eq.~\eqref{equation:3} & $\mathcal{L}_{\text{WM}}$ & \Checkmark \\
MROV & Eq.~\eqref{equation:10} & $\mathcal{L}_{\text{WM}}$ & \XSolidBrush \\
Certified robustness & Eq.~\eqref{equation:10} & Eq.~\eqref{equation:11} & \XSolidBrush \\
\hline
MROV-V-1 & Eq.~\eqref{equation:13} & $\mathcal{L}_{\text{WM}}$ & \Checkmark \\
MROV-V-2 & Eq.~\eqref{equation:13} & Eq.~\eqref{equation:14} & \Checkmark \\
\bottomrule
\end{tabular}}
\label{table:2}
\end{table}

Such a paradigm cannot be applied to black-box DNN watermarking schemes since the parameters within the victim model are hidden so $\mu$ is nowhere to be added.
In general cases, the judge resorts to the ownership key $K$, which is always available.
Concretely, a distribution $\pi'$ over the ownership key is explored as a substitution of $\pi$ such that the distribution over the embedding loss remains invariable, i.e., for any real number $t$:
\begin{footnotesize}
\begin{equation}
\label{equation:12}
\mathop{\text{Pr}}\limits_{\mu\leftarrow\pi}\!(\mathcal{L}_{\text{WM}}(M_{\text{WM}}\!+\!\mu,K\!,\textbf{m})\!\leq\!t)\!=\!\mathop{\text{Pr}}\limits_{\kappa\leftarrow\pi'}\!(\mathcal{L}_{\text{WM}}(M_{\text{WM}}\!,K\!+\!\kappa,\textbf{m})\!\leq\!t),
\end{equation}
\end{footnotesize}


The distribution $\pi'$ could be very complicated, especially for black-box schemes where the ownership key is a set of images or texts.
Therefore, we implement the distribution transfer by fitting $\pi'$ with a parameterized generator $G$. 

A collection of parameters pertubations $\left\{\mu_{q} \right\}_{q=1}^{Q}$ are randomly sampled from $\pi$. 
Then it is transformed into perturbations on the ownership key $\left\{\kappa_{q} \right\}_{q=1}^{Q}$ subject to:
\begin{equation}
\label{equation:add1}
\forall q\text{, }\mathcal{L}_{\text{WM}}(M_{\text{WM}}\!+\!\mu_{q},K\!,\!\mathbf{m})\!=\!\mathcal{L}_{\text{WM}}(M_{\text{WM}},K\!+\!\kappa_{q},\!\mathbf{m}). 
\end{equation}
A variational autoencoder~\cite{pu2016variational} is trained to reconstruct $\left\{\kappa_{q} \right\}_{q=1}^{Q}$, whose decoder is returned as $G$. 
The accuracy of this setting is established by the following theorem.

\begin{theorem}
\label{theorem:5}
$\pi'\!=\!\left\{G(\mathbf{z})\!:\!\mathbf{z}\!\leftarrow\!\mathcal{N}(0,\mathbf{I}) \right\}$ satisfies Eq.~\eqref{equation:12}.
\end{theorem}

\begin{proof}
Statistically, the l.h.s. of Eq.~\eqref{equation:12} equals the portion of samples $\mu$ satisfying $\mathcal{L}_{\text{WM}}(M\!+\!\mu,K\!,\mathbf{m})\!\leq\!t$.
With $\pi'\!=\!\left\{G(\mathbf{z})\!:\!\mathbf{z}\!\leftarrow\!\mathcal{N}(0,\mathbf{I}) \right\}$, the r.h.s. of Eq.~\eqref{equation:12} is reduced to:
\begin{footnotesize}
\begin{equation}
\nonumber
\begin{aligned}
&\int \text{Pr}(\mathbf{z})\cdot\mathbb{I}\left[(\mathcal{L}_{\text{WM}}(M_{\text{WM}},K\!+\!G(\mathbf{z}),\mathbf{m})\leq t\right]\text{d} \mathbf{z}\\
\!=\!&\frac{|\!\left\{q\!:\!\mathcal{L}_{\text{WM}}(M_{\text{WM}}\!,\!K\!+\!\kappa_{q},\!\mathbf{m})\!\leq\! t\right\}\!|}{Q}\!=\!\frac{|\!\left\{q\!:\!\mathcal{L}_{\text{WM}}(M_{\text{WM}}\!+\!\mu_{q},\!K,\!\mathbf{m})\!\leq\! t\right\}\!|}{Q},
\end{aligned}
\end{equation}
\end{footnotesize}
where the second equation holds since $G$ is trained by assuming its inputs as a normal distribution, and the last equation holds by Eq.~\eqref{equation:add1}.
\end{proof}

Having obtained the ownership key perturbation generator $G$, the message is retrieved by:
\begin{footnotesize}
\begin{equation}
\label{equation:13}
\hat{\mathbf{m}}[i]=\arg\max_{b\in\left\{0,1\right\}}\left\{\sum_{r=1}^{R}\mathbb{I}\left[\texttt{Verify}(M,\!K\!+\!\kappa_{r})[i]=b\right]\right\}.
\end{equation}
\end{footnotesize}
Theoretically, Eq.~\eqref{equation:13} is the variational approximation version of Eq.~\eqref{equation:10} so the results are expected to be similar if $G$ accurately captures the influence of $\pi$. 
Since $G$ depends on a watermarked model, the owner can watermark its model, run the variational distribution transfer program, and further fine-tune the watermarked model with the following regularizer to enhance the robustness:
\begin{equation}
\label{equation:14}
\mathcal{L}_{\text{WM}}^{K}(M_{\text{WM}},K\!,\!\mathbf{m})\!=\!\frac{1}{P}\sum_{p=1}^{P}\mathcal{L}_{\text{WM}}(M_{\text{WM}},\!K\!+\!\kappa_{p},\!\mathbf{m}). 
\end{equation}
 
Different configurations of the watermarking scheme regarding the accurate transmission of the identity message are summarized in Table~\ref{table:2} and Fig.~\ref{figure:4}. 

\begin{figure}[!tbp]
\centering
\subfigure[MROV.]{
\begin{minipage}[htbp]{0.5\linewidth}
\centering
\includegraphics[width=3.6cm,height=3cm]{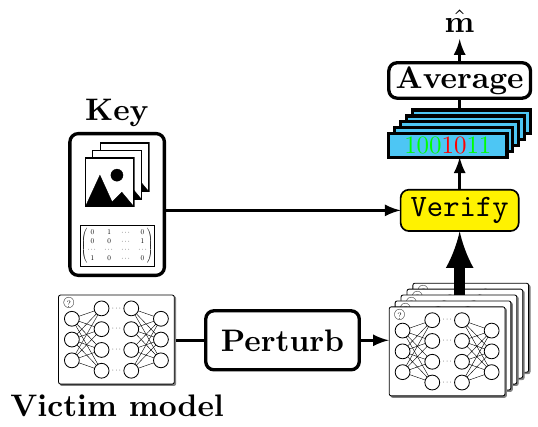}
\end{minipage}
}\subfigure[MROV-V.]{
\begin{minipage}[htbp]{0.5\linewidth}
\centering
\includegraphics[width=4cm,height=3cm]{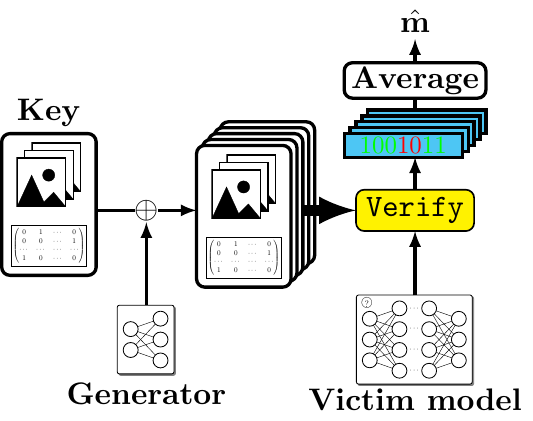}
\end{minipage}
}
\caption{Multiple rounds of ownership verification, conducted by the judge.
(a) can only be applied to white-box schemes. 
(b) can be applied to any scheme. 
}
\label{figure:4}
\end{figure}


\section{Experiments and Discussions}
\label{sec:4}
\subsection{Settings}
We studied the capacity properties of several representative DNN watermarking schemes. 
\textbf{Uchida} et~al.'s scheme~\cite{main:38} is a prototypical static parameter-based white-box watermarking scheme. 
The identity message is explicitly defined as Eq.~\eqref{equation:4}. 
Spread-Transform Dither Modulation watermarking \textbf{(STDM)}~\cite{main:42} is a variant of Uchida et~al.'s scheme with STDM activation function. 
\textbf{MTLSign}~\cite{main:46} is a dynamic white-box scheme. 
Each bit of the identity message is retrieved from the prediction for a trigger returned from a binary classifier based on hidden neurons' responses. 
\textbf{Content} is a representative pattern of triggers defined in Zhang et~al.'s black-box scheme~\cite{main:8}. 
\textbf{Exponential}-weighting~\cite{main:11} uses normal training sample as triggers with exponential regularizer during watermark embedding. 
\textbf{Frontier}-stitching~\cite{main:9} uses samples close to the decision boundary with adversarial perturbations as triggers. 

\begin{figure}[!t]
\centering
\subfigure[Uchida.]{
\begin{minipage}[htbp]{0.5\linewidth}
\centering
\includegraphics[width=4cm,height=2cm]{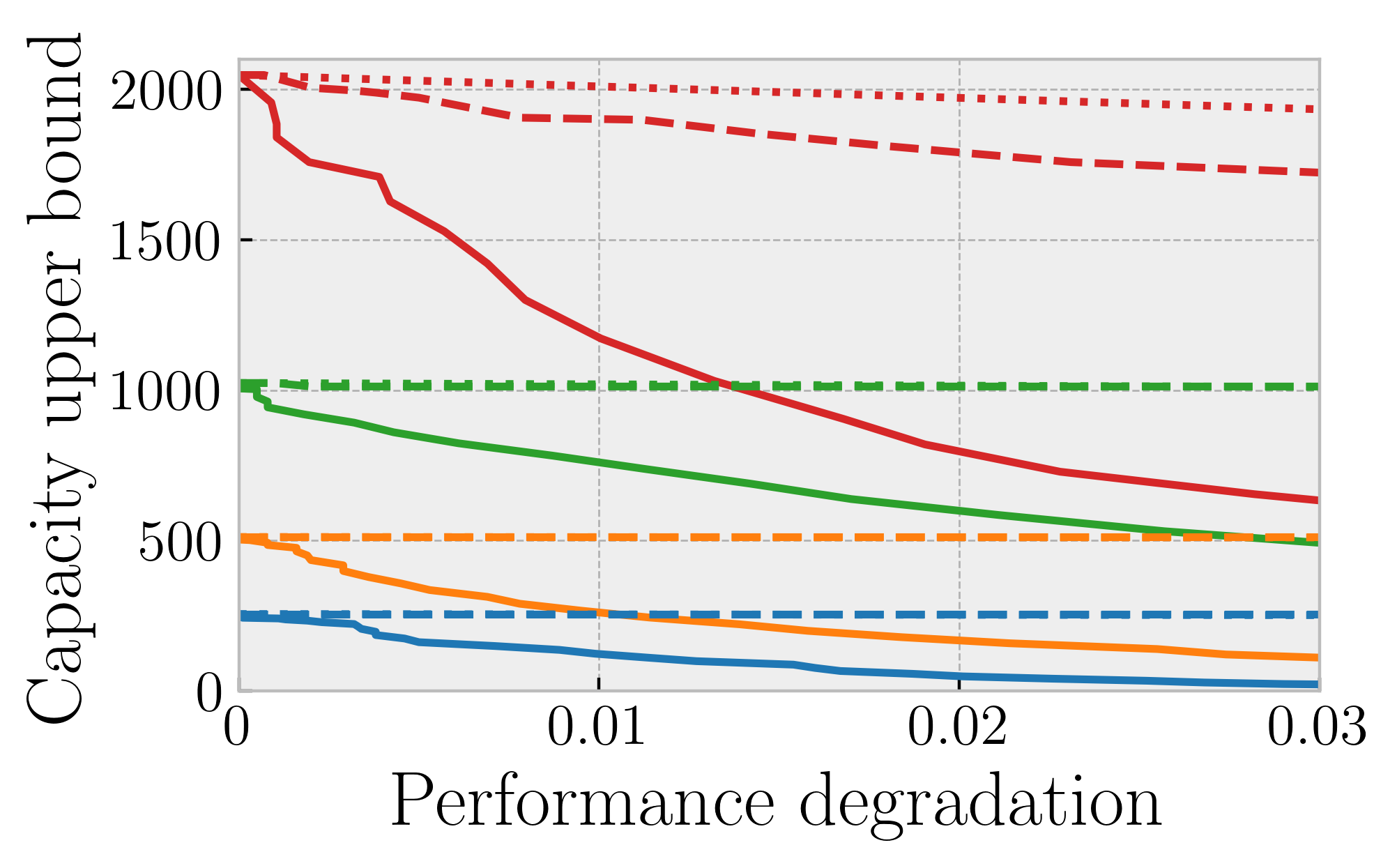}
\end{minipage}
}\subfigure[STDM.]{
\begin{minipage}[htbp]{0.5\linewidth}
\centering
\includegraphics[width=4cm,height=2cm]{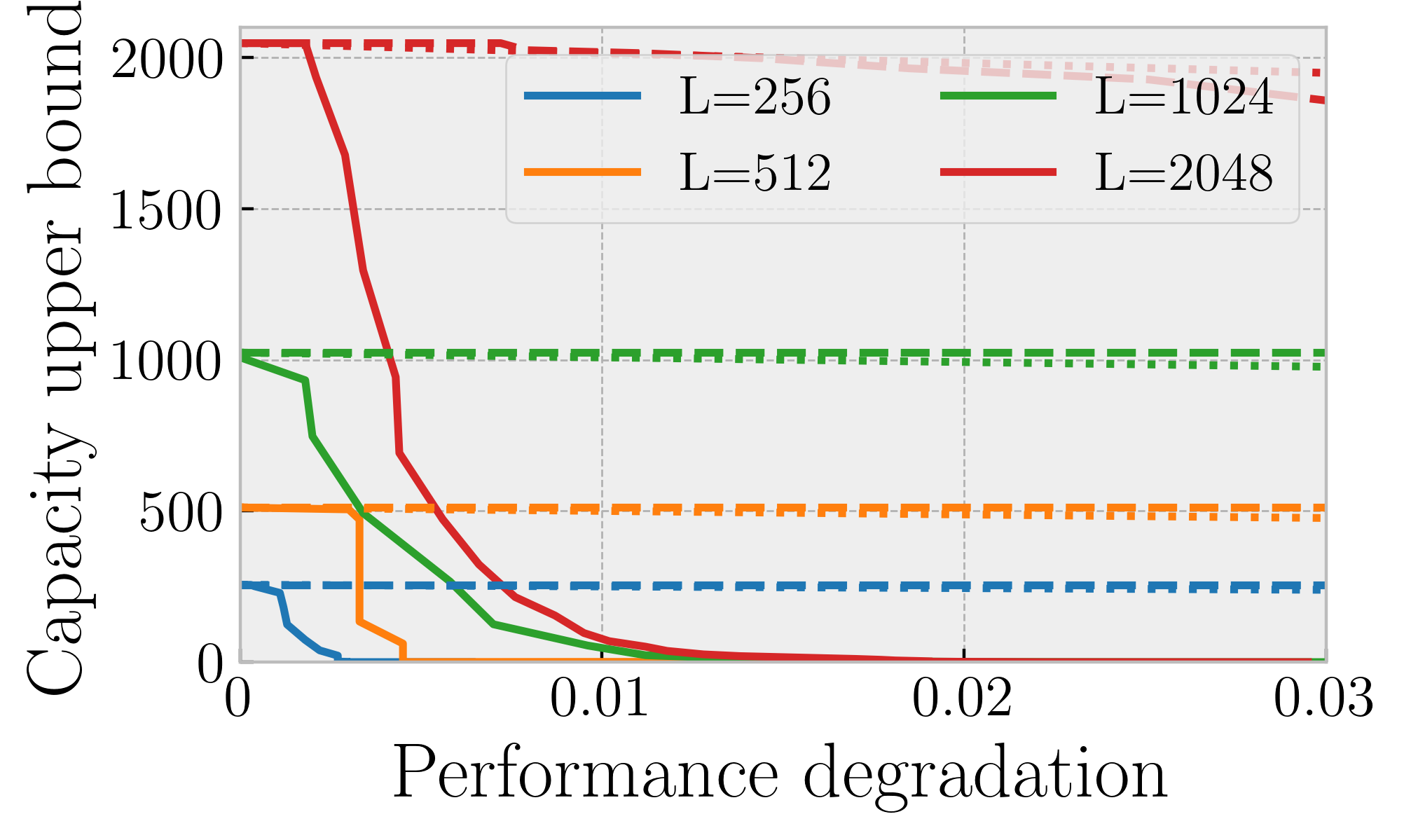}
\end{minipage}
}
\subfigure[MTLSign.]{
\begin{minipage}[htbp]{0.5\linewidth}
\centering
\includegraphics[width=4cm,height=2cm]{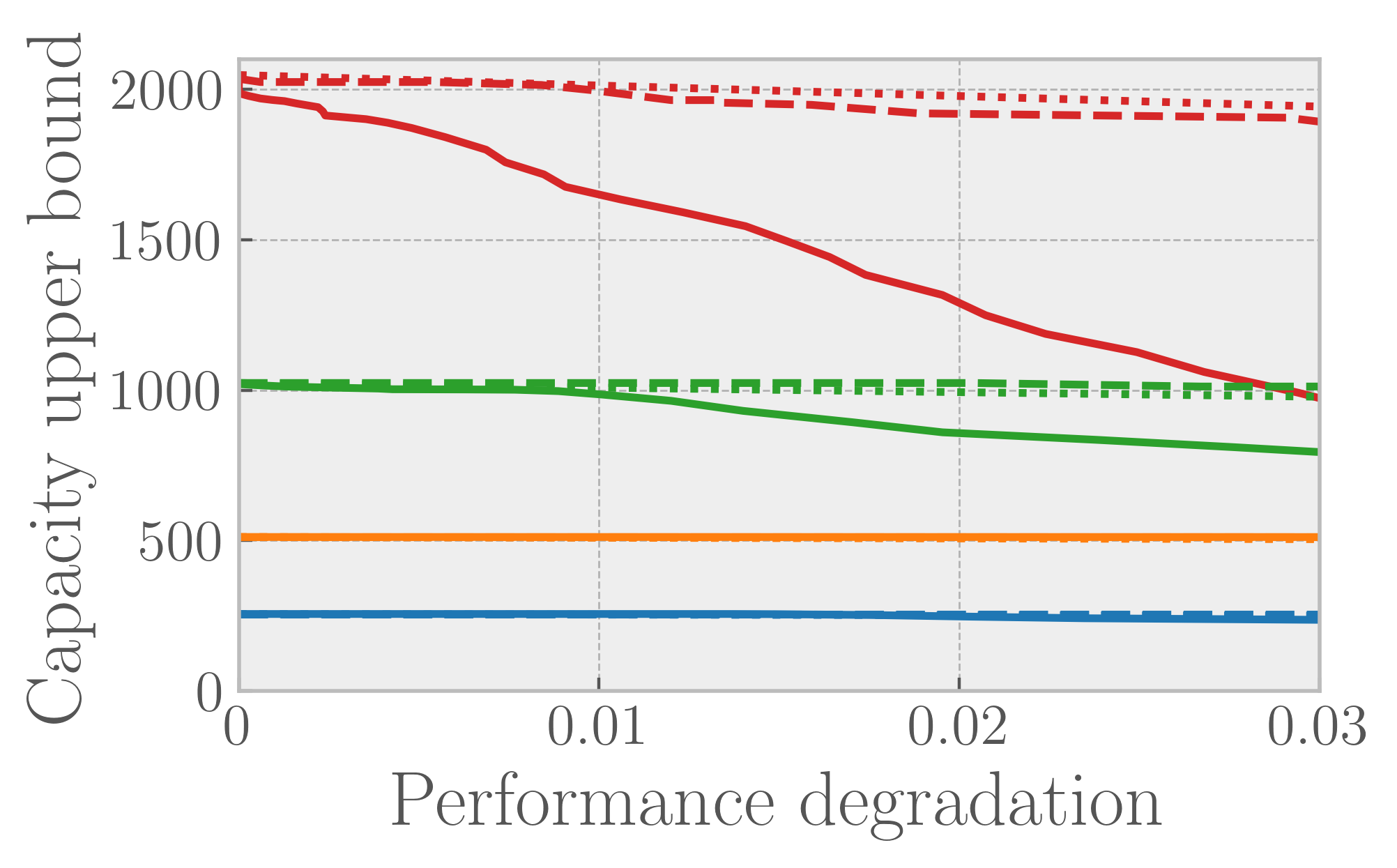}
\end{minipage}
}\subfigure[Content.]{
\begin{minipage}[htbp]{0.5\linewidth}
\centering
\includegraphics[width=4cm,height=2cm]{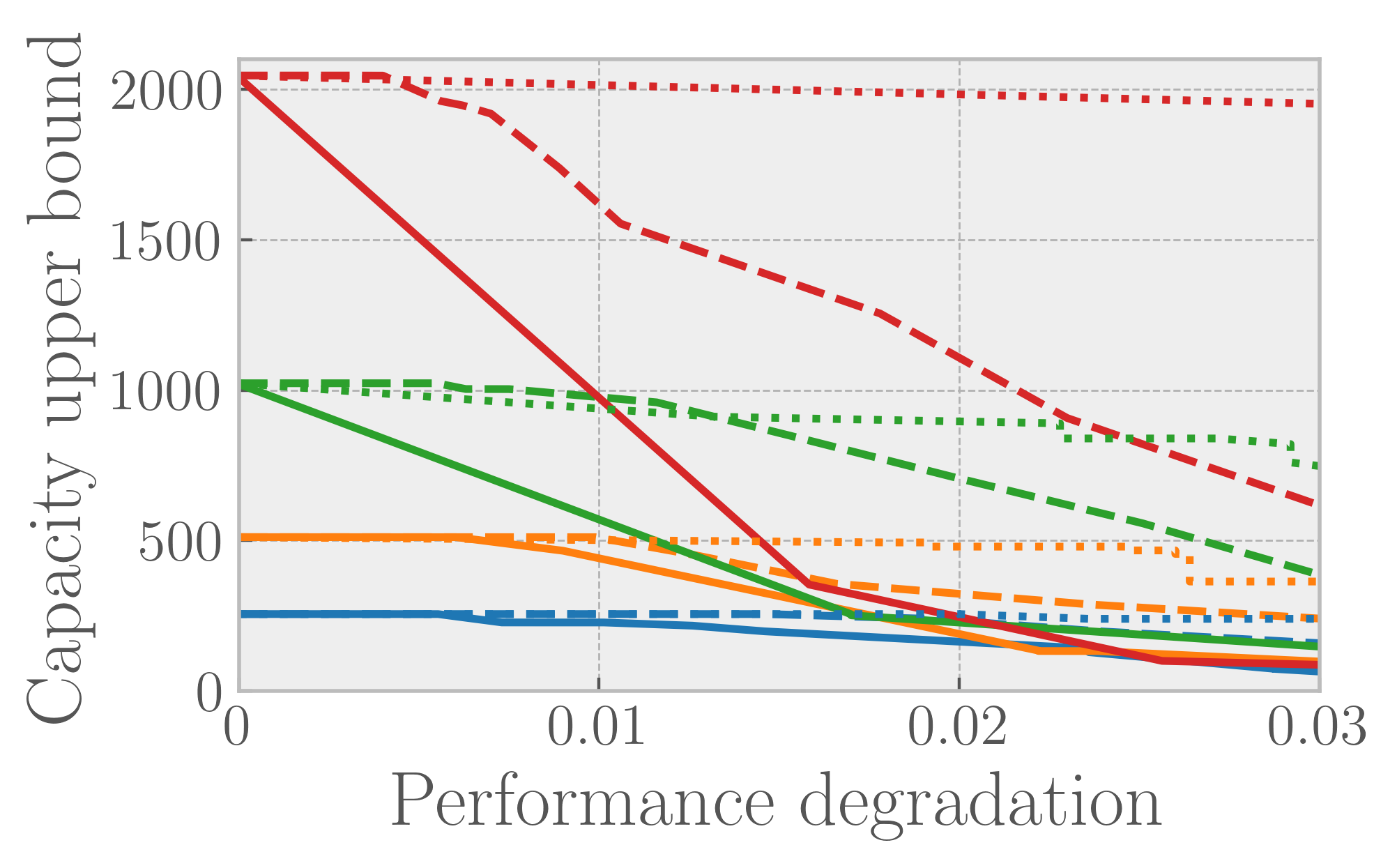}
\end{minipage}
}
\subfigure[Exponential.]{
\begin{minipage}[htbp]{0.5\linewidth}
\centering
\includegraphics[width=4cm,height=2cm]{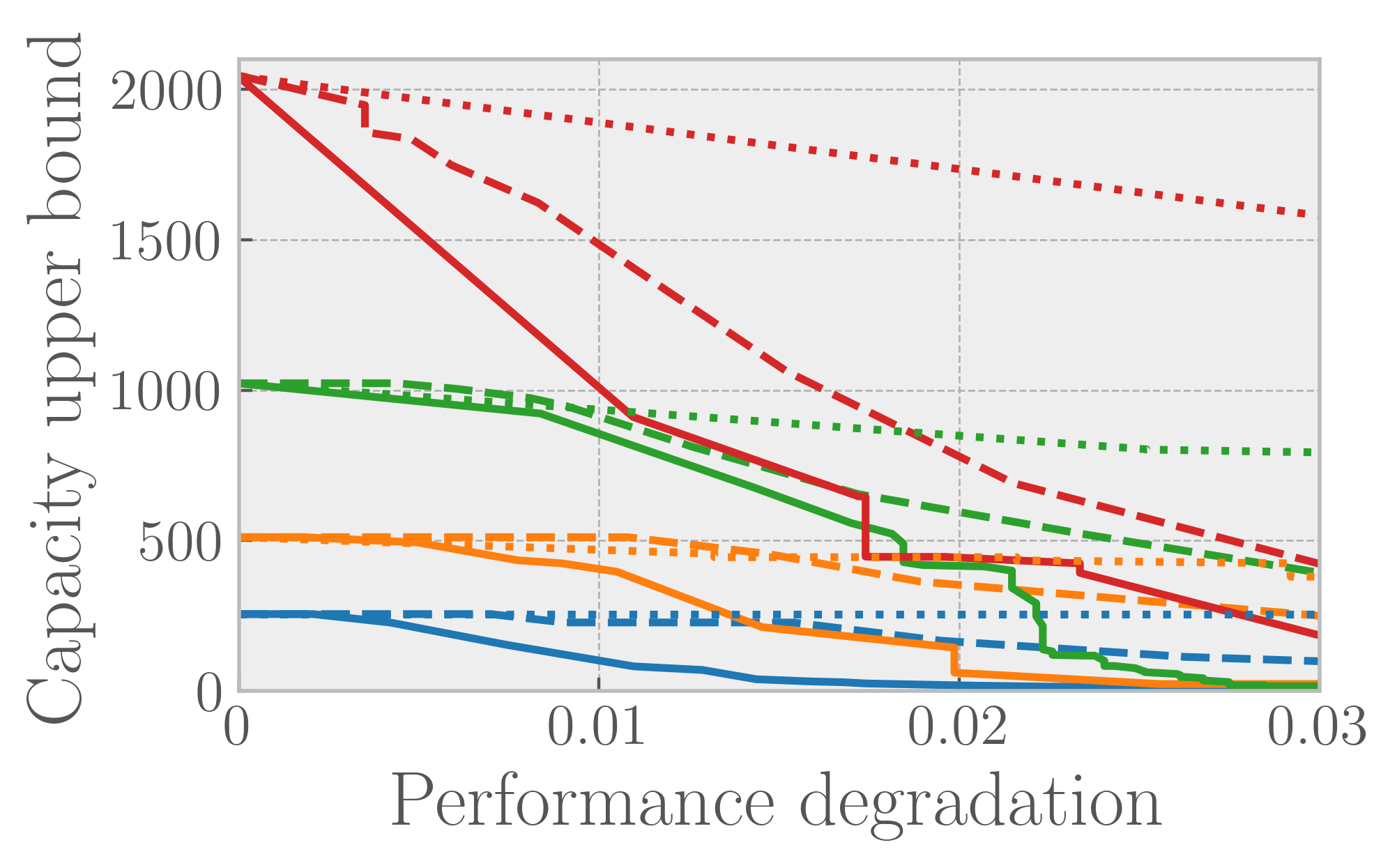}
\end{minipage}
}\subfigure[Frontier.]{
\begin{minipage}[htbp]{0.5\linewidth}
\centering
\includegraphics[width=4cm,height=2cm]{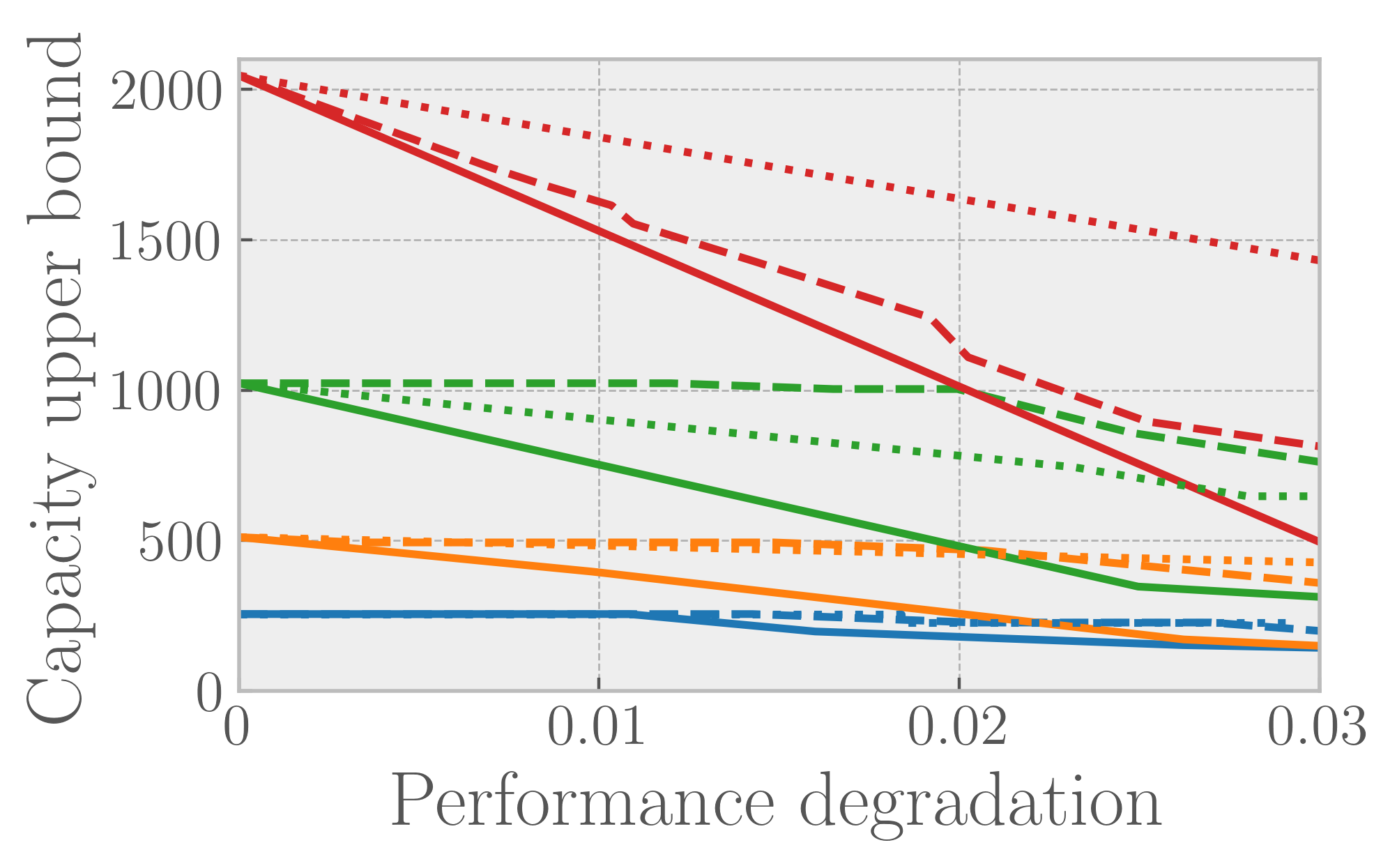}
\end{minipage}
}
\caption{Estimated capacity $\hat{C}(\delta,L)$ as a function of the performance degradation $\delta$ under fine-tuning (\protect\tikz[baseline]{\protect\draw[line width=1pt,dotted](0,0.5ex)--(0.5,0.5ex);}), 
neuron-pruning (\protect\tikz[baseline]{\protect\draw[line width=1pt,dashed](0,0.5ex)--(0.5,0.5ex);}), and 
adversarial overwriting (\protect\tikz[baseline]{\protect\draw[line width=1pt](0,0.5ex)--(0.5,0.5ex);}). 
The length of the identity message $L$ is set as {\color{RoyalBlue}256}, {\color{Orange}512}, {\color{LimeGreen}1024}, and {\color{Crimson}2048}. 
}
\label{figure:42}
\end{figure}

The task is classification on CIFAR-10~\cite{krizhevsky2009learning}. 
The architecture of DNN model to be protected is ResNet-50~\cite{he2016deep}. 
To compile the identity message into classification predictions, we adopted the vanilla interpreter and mapped the prediction of each trigger into $\lfloor\log_{2} 10 \rfloor\!=\!3$ bits. 
All experiments were implemented in \texttt{PyTorch} framework~\footnote{Source codes are available upon request.}.

\begin{table*}[!t]
\centering
\caption{Expense of copyright protection measured by the necessarily minimal length of the identity message (a) and performance degradation due to watermarking (b) computed by Theorem~\ref{theorem:3} when the ownership can be accurately retrieved from variants of the watermarked model with performance degradation no more than $\Delta$. 
\textbf{Bold entries} are optimal configurations. 
$J\!=\!1024$. 
}
\subtable[Minimal length of the identity message $\tilde{L}$ (bit).]{
\scalebox{0.65}{
\begin{tabular}{c|c|c|c|c|c|c}
\toprule
\multirow{2}{*}{$\Delta$} & \multicolumn{6}{c}{\textbf{Watermarking scheme}} \\
\cline{2-7}
& Uchida & STDM & MTLSign & Content & Exponential & Frontier \\
\hline
1\% & 1600$\pm$50 & 3900$\pm$300 & \textbf{1300}$\pm$\textbf{100} & 1500$\pm$50 & 1800$\pm$200 & 1450$\pm$50 \\
2\% & 2100$\pm$200 & $\geq$8192 & \textbf{1600}$\pm$\textbf{200} & 1900$\pm$50 & 2800$\pm$800 & 1850$\pm$50 \\
3\% & 3600$\pm$600 & $\geq$8192 & \textbf{1900}$\pm$\textbf{300} & 4100$\pm$600 & 5000$\pm$100 & 2800$\pm$300 \\
\bottomrule
\end{tabular}}}%
\subtable[Minimal performance degradation due to watermarking $F( \tilde{L})$ (\%).]{
\scalebox{0.65}{
\begin{tabular}{c|c|c|c|c|c|c}
\toprule
\multirow{2}{*}{$\Delta$} & \multicolumn{6}{c}{\textbf{Watermarking scheme}} \\
\cline{2-7}
& Uchida & STDM & MTLSign & Content & Exponential & Frontier \\
\hline
1\% & 0.296$\pm$0.003 & 0.376$\pm$0.029 & 0.162$\pm$0.012 & 0.130$\pm$0.004 & 0.156$\pm$0.017 & \textbf{0.124$\pm$0.004} \\
2\% & 0.326$\pm$0.012 & $\geq$0.790 & 0.199$\pm$0.025 & 0.165$\pm$0.004 & 0.243$\pm$0.069 & \textbf{0.158$\pm$0.004} \\
3\% & 0.415$\pm$0.030 & $\geq$0.790 & \textbf{0.237$\pm$0.037} & 0.355$\pm$0.052 & 0.433$\pm$0.009 & 0.249$\pm$0.026 \\
\bottomrule
\end{tabular}}}
\label{table:42}
\end{table*}

\subsection{Capacity estimation}
The dataset with which the adversary conducts the attacks was 10\% of the entire training dataset. 
$\lambda$ in adversarial overwriting formulated in Table~\ref{table:attacks} was set as $0.1$. 
Fine-tuning and neuron-pruning were implemented as in~\cite{main:5}. 
Each configuration was repeated for five times, the mean capacity estimation of all watermarking schemes by Algo.~\ref{algorithm:1} is illustrated in Fig.~\ref{figure:42}. 
We made the following observations.  

(i) In general, $\hat{C}(\delta,L)$ declines in $\delta$ and increases in $L$ as predicted by Theorem~\ref{theorem:1}. 
(ii) Compared with universal watermark removal attacks with weaker assumptions on the adversary such as fine-tuning and neuron-pruning, adversarial overwriting yielded the lowest and tightest estimation of the capacity since it directly tampers with the watermark while preserves the model's performance simultaneously. 
(iii) Capacity varies with the watermarking scheme. 
As a covert version of \textbf{Uchida}, the capacity of \textbf{STDM} is extremely small since the embedded information be removed with slight modification and hence negligible performance degradation. 
This fact reflects the tradeoff between capacity and stealthiness. 
The capacity of black-box schemes depends on the pattern of triggers and has to be measured a posteriori. 

As a conclusive evaluation of studied DNN watermarking schemes, we applied Theorem~\ref{theorem:3} and computed the minimal length of the identity message and the minimal performance degradation due to watermark embedding. 
We set $J\!=\!1024$ as a standard digital signature scheme DSA-1024 with $\Delta\!=\!\left\{1\%,2\%,3\% \right\}$. 
As demonstrated in Fig.~\ref{figure:423}, the monotonicity is identical to predictions in Fig.~\ref{figure:2}, the fidelity $F(L)$ was computed and interpolated for $L=\left\{256,512,1024,2048,4096,8192 \right\}$. 
$\tilde{L}$ was estimated using Eq.~\eqref{equation:9} with a granularity of 50 bit, numerical results are given in Table~\ref{table:42}. 
We noted that the scheme with the largest capacity (\textbf{MTLSign}) had the smallest minimal identity message length, but the corresponding performance degradation is not always the lowest. 
Therefore, a good watermarking scheme should have both a large capacity and a high fidelity.


\begin{figure}[!tbp]
\centering
\includegraphics[width=5.2cm,height=3.9cm]{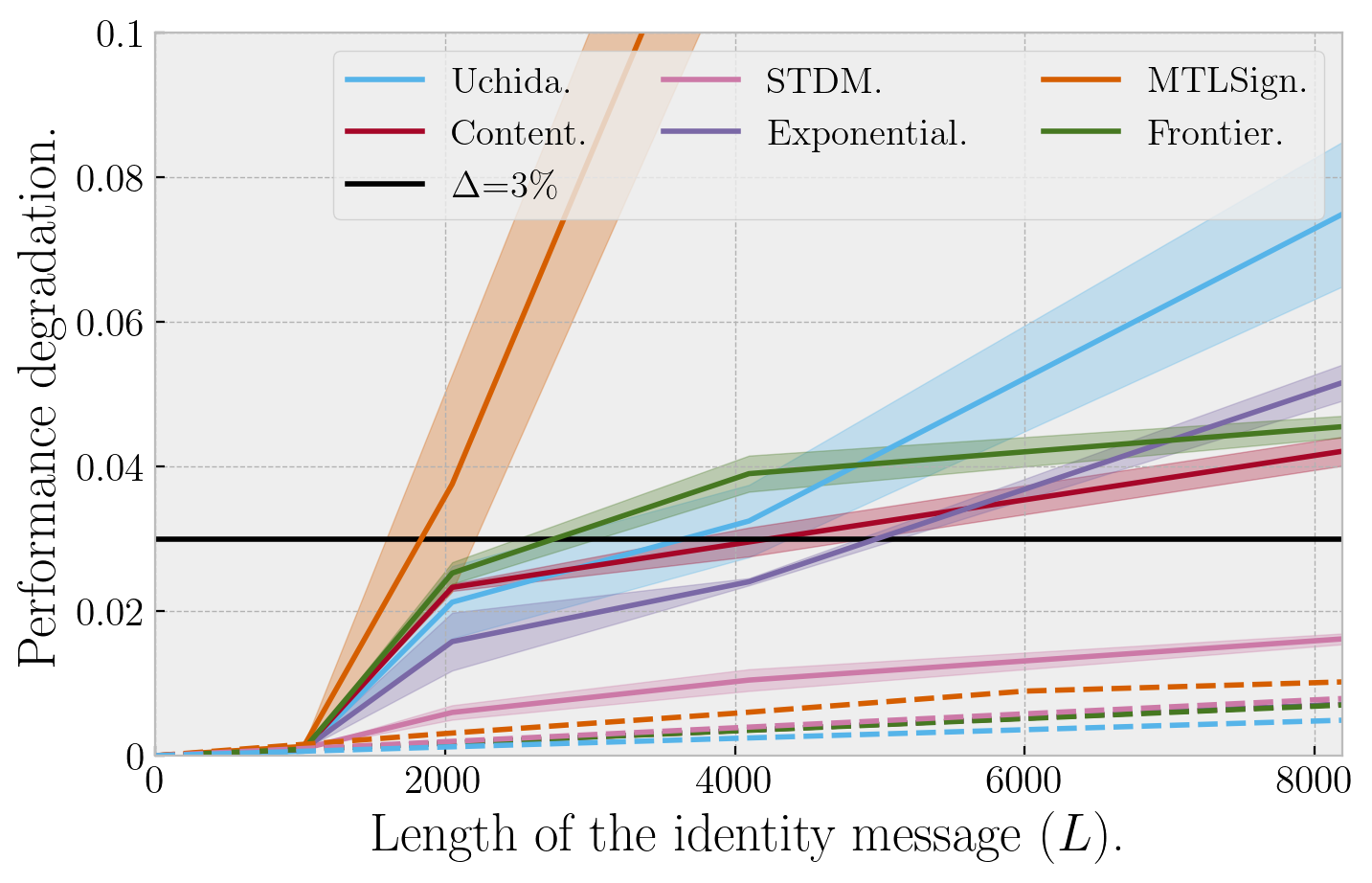}
\caption{Estimated minimal length of the identity message, $J\!=\!1024$. 
The dashed lines mark the cost in fidelity $F(L)$. 
The solid lines mark  $F(L)\!+\!\min_{\delta}$\{$\delta\!:\!\hat{C}(\delta,L)\!\leq\! J$\}, shadow areas denote fluctuations during experiments.
}
\label{figure:423}
\end{figure}

\begin{table*}[!t]
\centering
\caption{Expense of copyright protection using DNN watermarking schemes computed by Theorem~\ref{theorem:3} when the ownership can be accurately retrieved from variants of the watermarked model with performance degradation no more than $\Delta$. $J=1024$. 
\colorbox{gray!25}{Entries marked in shadow} are configurations outperformed by the baseline in Table~\ref{table:42}. 
}
\subtable[Minimal length of the identity message $\tilde{L}$ (bit).]{
\scalebox{0.57}{
\begin{tabular}{c|c|c|c|c|c|c|c}
\toprule
\multirow{2}{*}{$\Delta$} & \multirow{2}{*}{\textbf{Configuration}} & \multicolumn{6}{c}{\textbf{Watermarking scheme}} \\
\cline{3-8}
& & Uchida & STDM & MTLSign & Content & Exponential & Frontier \\
\hline
\multirow{2}{*}{1\%} & MROV-V-1 & 1350$\pm$50 & 2500$\pm$200 & \textbf{1250$\pm$50} & 1400$\pm$50 & 1650$\pm$150 & 1350$\pm$50 \\
& MROV-V-2 & \textbf{1250$\pm$50} & 2400$\pm$50 & \textbf{1250$\pm$50} & 1300$\pm$50 & 1450$\pm$100 & \textbf{1250$\pm$50} \\
\hline
\multirow{2}{*}{2\%} & MROV-V-1 & 1600$\pm$50 & 4200$\pm$400 & 1500$\pm$50 & 1800$\pm$100 & 2700$\pm$800 & 1650$\pm$50 \\
& MROV-V-2 & 1500$\pm$50 & 4250$\pm$150 & \textbf{1450$\pm$50} & 1600$\pm$100 & 2100$\pm$300 & 1500$\pm$50 \\
\hline
\multirow{2}{*}{3\%} & MROV-V-1 & 1800$\pm$50 & 7400$\pm$400 & 1750$\pm$50 & 2500$\pm$300 & 4400$\pm$600 & 2050$\pm$50 \\
& MROV-V-2 & 2000$\pm$50 & 6000$\pm$400 & \textbf{1700$\pm$50} & 2100$\pm$300 & 2700$\pm$500 & 1800$\pm$50 \\

\bottomrule
\end{tabular}}}%
\subtable[Minimal performance degradation due to watermarking $F( \tilde{L})$ (\%).]{
\scalebox{0.57}{
\begin{tabular}{c|c|c|c|c|c|c|c}
\toprule
\multirow{2}{*}{$\Delta$} & \multirow{2}{*}{\textbf{Configuration}} & \multicolumn{6}{c}{\textbf{Watermarking scheme}} \\
\cline{3-8}
& & Uchida & STDM & MTLSign & Content & Exponential & Frontier \\
\hline
\multirow{2}{*}{1\%} & MROV-V-1 & 0.168$\pm$0.005 & 0.241$\pm$0.019 & 0.156$\pm$0.006 & 0.121$\pm$0.004 & 0.143$\pm$0.013 & \textbf{0.115$\pm$0.004} \\
& MROV-V-2 & 0.233$\pm$0.006 & 0.417$\pm$0.009 & \cellcolor{gray!25}0.233$\pm$0.009 & \cellcolor{gray!25}0.236$\pm$0.010 & \cellcolor{gray!25}0.251$\pm$0.017 & \cellcolor{gray!25}0.235$\pm$0.011 \\
\hline
\multirow{2}{*}{2\%} & MROV-V-1 & 0.199$\pm$0.005 & 0.405$\pm$0.039 & 0.187$\pm$0.006 & 0.156$\pm$0.009 & 0.234$\pm$0.069 & \textbf{0.141$\pm$0.004} \\
& MROV-V-2 & 0.280$\pm$0.008 & 0.738$\pm$0.026 & \cellcolor{gray!25}0.271$\pm$0.010 & \cellcolor{gray!25}0.291$\pm$0.018 & \cellcolor{gray!25}0.364$\pm$0.052 & \cellcolor{gray!25}0.282$\pm$0.010 \\
\hline
\multirow{2}{*}{3\%} & MROV-V-1 & 0.249$\pm$0.006 & 0.714$\pm$0.039 & \textbf{0.226$\pm$0.006} & 0.281$\pm$0.026 & 0.381$\pm$0.052 & 0.245$\pm$0.004 \\
& MROV-V-2 & 0.336$\pm$0.009 & 1.042$\pm$0.069 & \cellcolor{gray!25}0.318$\pm$0.009 & \cellcolor{gray!25}0.382$\pm$0.055 & \cellcolor{gray!25}0.468$\pm$0.087 & \cellcolor{gray!25}0.338$\pm$0.011 \\
\bottomrule
\end{tabular}}}
\label{table:43}
\end{table*}

\begin{figure}[!t]
\centering
\subfigure[Uchida.]{
\begin{minipage}[htbp]{0.5\linewidth}
\centering
\includegraphics[width=4cm,height=2cm]{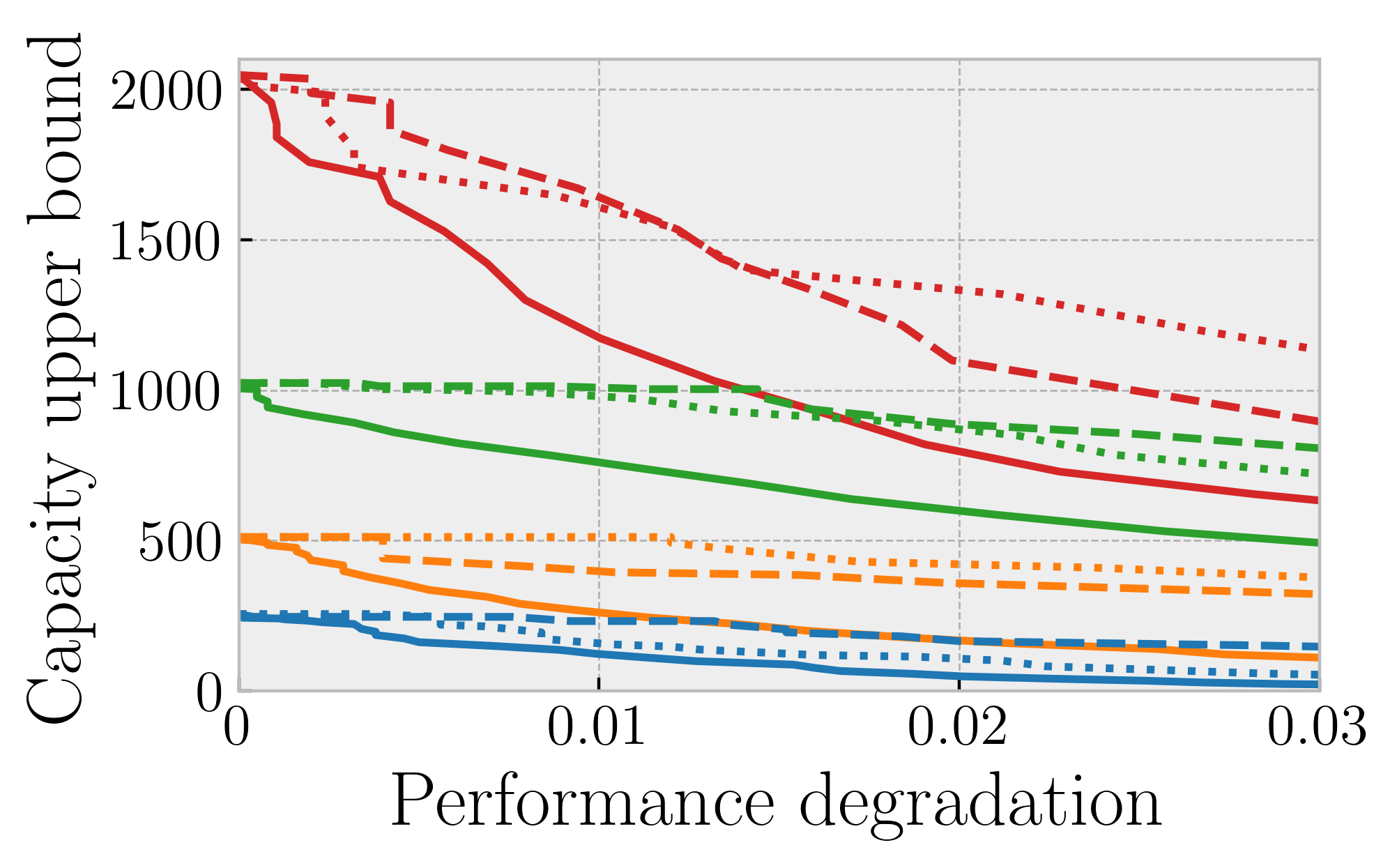}
\end{minipage}
}\subfigure[STDM.]{
\begin{minipage}[htbp]{0.5\linewidth}
\centering
\includegraphics[width=4cm,height=2cm]{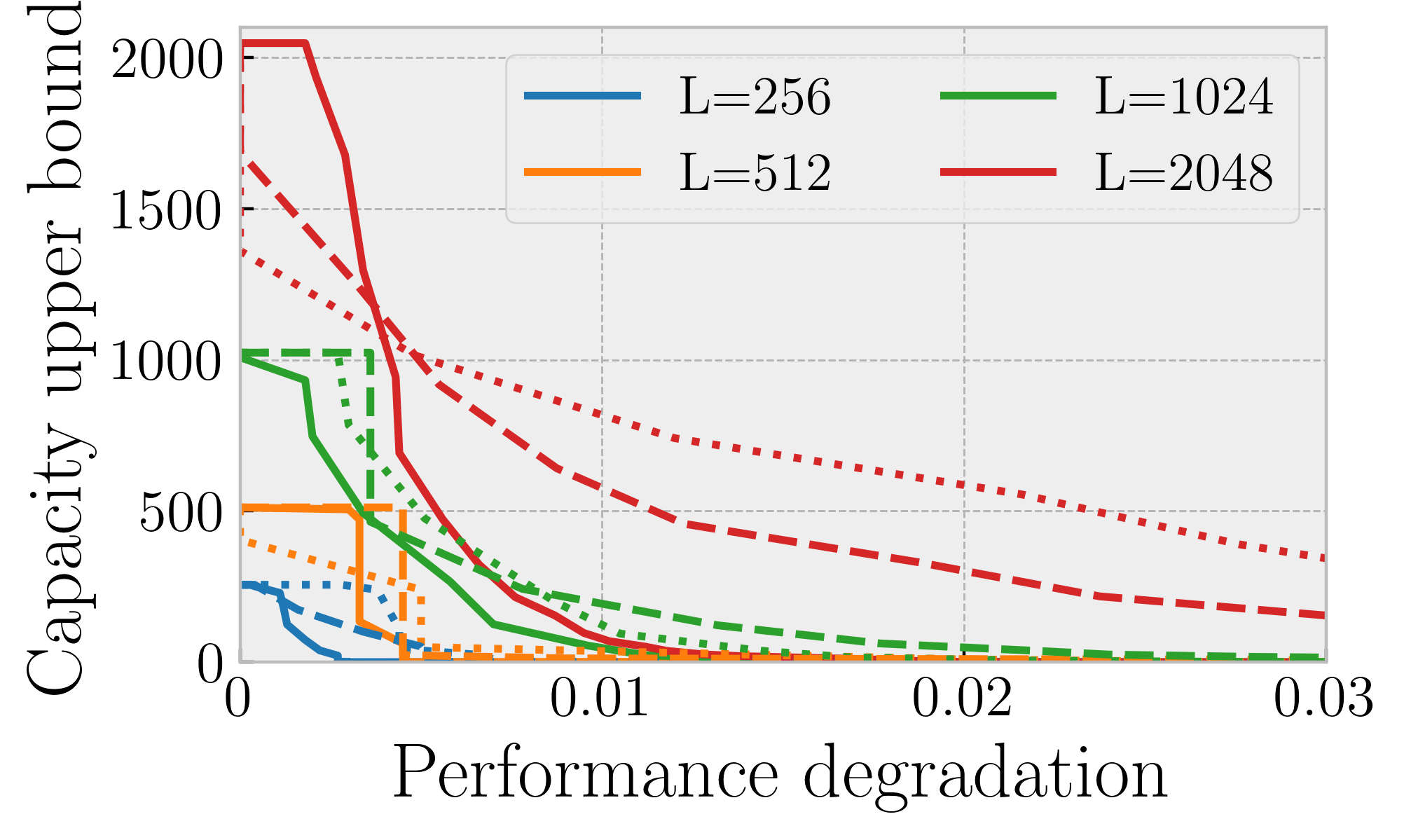}
\end{minipage}
}
\subfigure[MTLSign.]{
\begin{minipage}[htbp]{0.5\linewidth}
\centering
\includegraphics[width=4cm,height=2cm]{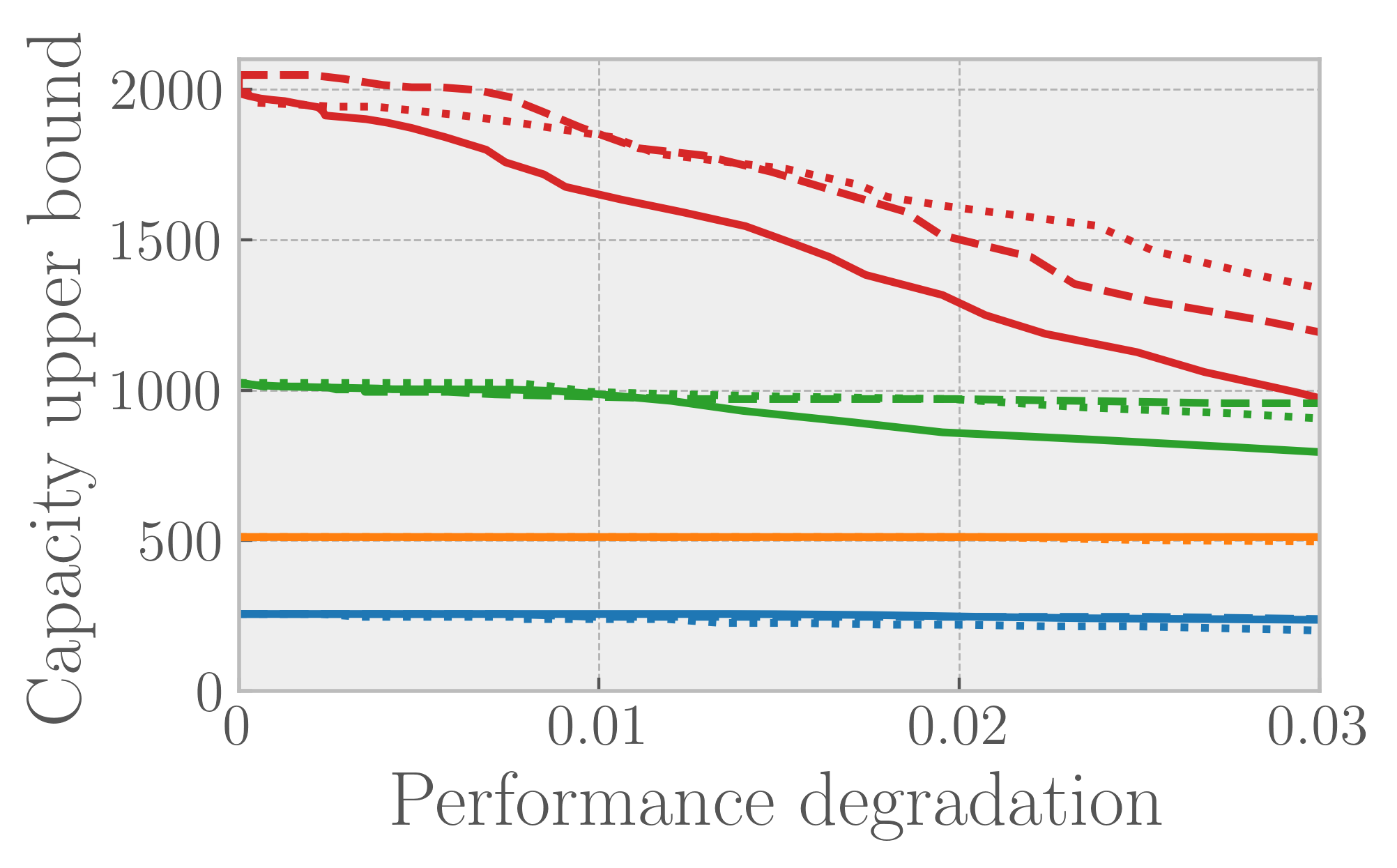}
\end{minipage}
}\subfigure[Content.]{
\begin{minipage}[htbp]{0.5\linewidth}
\centering
\includegraphics[width=4cm,height=2cm]{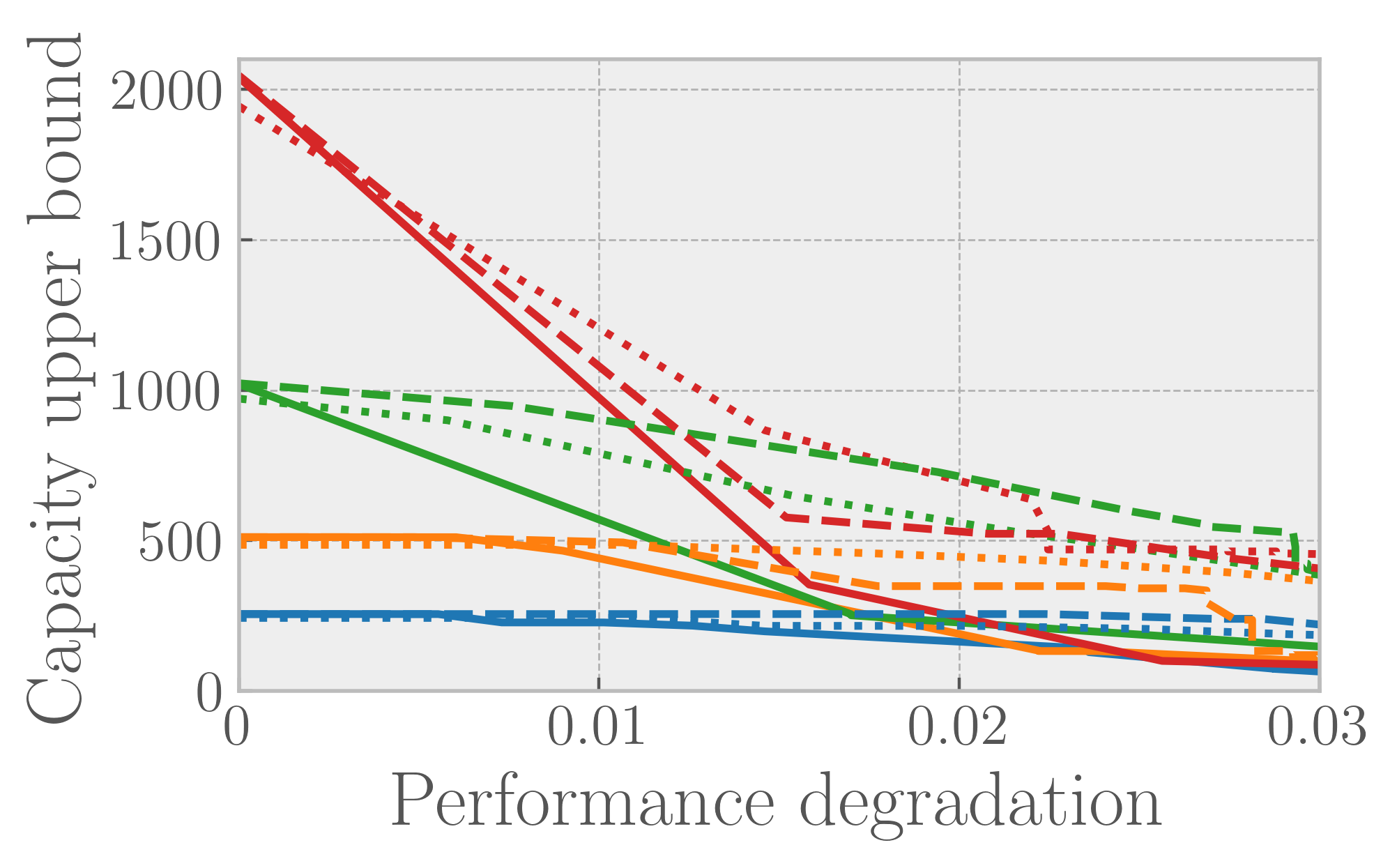}
\end{minipage}
}
\subfigure[Exponential.]{
\begin{minipage}[htbp]{0.5\linewidth}
\centering
\includegraphics[width=4cm,height=2cm]{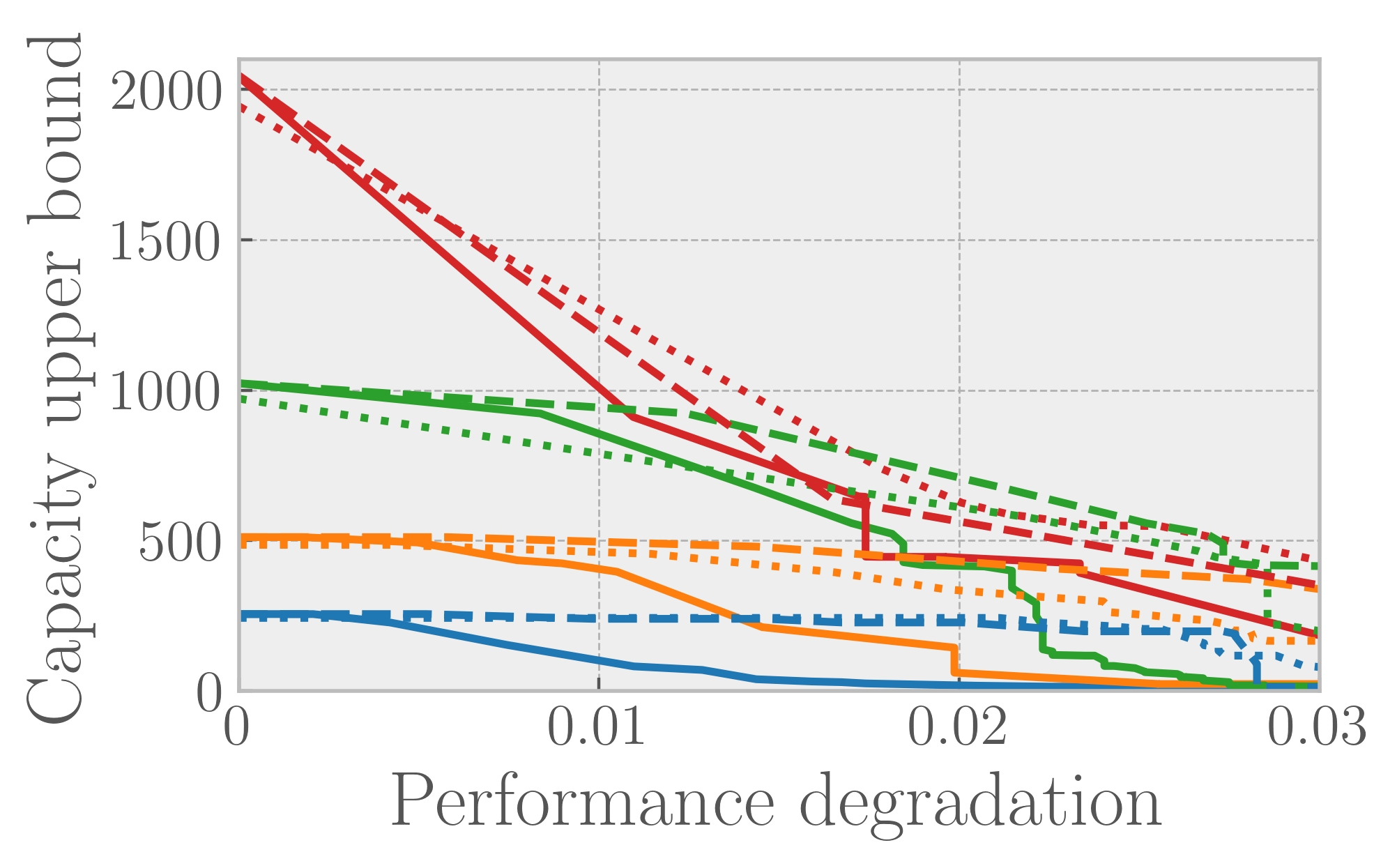}
\end{minipage}
}\subfigure[Frontier.]{
\begin{minipage}[htbp]{0.5\linewidth}
\centering
\includegraphics[width=4cm,height=2cm]{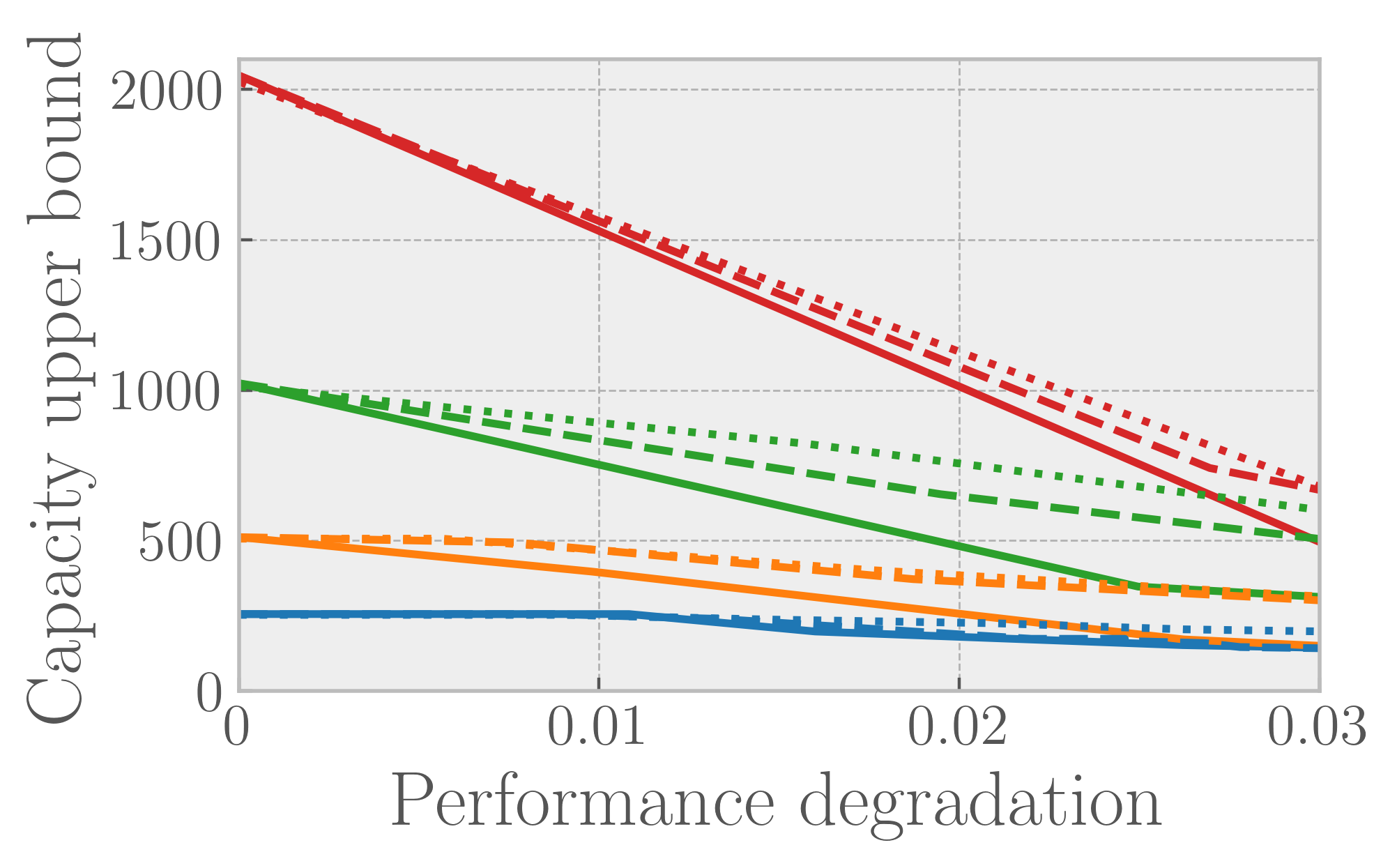}
\end{minipage}
}
\caption{Estimated capacity under adversarial overwriting under the baseline setting (\protect\tikz[baseline]{\protect\draw[line width=1pt](0,0.5ex)--(0.5,0.5ex);}), MROV-V-1(\protect\tikz[baseline]{\protect\draw[line width=1pt,dashed](0,0.52ex)--(0.5,0.52ex);}), and MROV-V-2 (\protect\tikz[baseline]{\protect\draw[line width=1pt,dotted](0,0.5ex)--(0.5,0.5ex);}). 
$L$ is set as {\color{RoyalBlue}256}, {\color{Orange}512}, {\color{LimeGreen}1024}, and {\color{Crimson}2048}.
}
\label{figure:43}
\end{figure}

\subsection{Efficacy of MROV-V}
We verified the efficacy of the proposed variational multiple rounds of ownership verification (we focused on MROV-V-1 and MROV-V-2 that are applicable in the black-box settings) for six watermarking schemes.

The distribution of perturbations on the parameters $\pi$ was set as a normal distribution as in~\cite{main:17}. 
The variational approximation distribution $\pi'$ for trigger-based schemes including \textbf{MTLSign}, \textbf{Content}, \textbf{Exponential}, and \textbf{Frontier} was generated by a decoder with four deconvolutional layers. 
For \textbf{Uchida} and \textbf{STDM}, the decoder was a four-layer MLP. 
Each autoencoder was trained on a collection of $Q\!=\!10000$ instances of parameters perturbations. 
The number of rounds was fixed as $R\!=\!100$. 
During watermark embedding of MROV-V-2 defined by Eq.~\eqref{equation:14}, $P\!=\!10$. 

Empirically, both MROV-V-1 and MROV-V-2 reduced the BER during ownership verification under adversarial modifications. 
To provide a fair and uniform comparison between MROV-V-1/2 and the basic one-time ownership verification, we transformed the BER statistics into capacity upper bounds using Theorem~\ref{theorem:2} and visualized them in Fig.~\ref{figure:43}. 
It is observed that MROV-V-1/2 increases the bound so more information can be safely transmitted through DNN watermark. 
This is because that the adversary has to modify the victim model until all neighbours of the ownership key fail to expose the correct identiy message (instead of only the ownership key as in one-time verification). 
MROV-V-2 has the largest capacity since the adversary has to exert larger modifications to suppress the effect of regularizer in Eq.~\eqref{equation:14}. 

However, MROV-V-2 also introduces a worse fidelity, i.e., although $\min_{\delta}\left\{\delta:C(\delta,L)\leq J \right\}$ decreases, $F(L)$ also increases so $\tilde{L}$ computed by Eq.~\eqref{equation:9} and $F(\tilde{L})$ are not guaranteed to decline. 
This fact is demonstrated by Fig.~\ref{figure:432}. 
Numerically, we computed the minimal length of the identity message and the corresponding performance degradation in Table~\ref{table:43}. 
We remark that both types of MROV-V reduce the minimal length of the identity message, but only MROV-V-1 can always reduce the minimal performance degradation due to watermarking since $F(L)$ is left invariable. 
In all settings of \textbf{MTLSign}, \textbf{Content}, \textbf{Exponential}, and \textbf{Frontier}, MROV-V-2 yields a higher cost because of the corrupted fidelity. 

As a result, we recommend the configuration of MROV-V-1 which is universally applicable, non-invasive, and is promised to increase the capacity's upper bound and hence reduce the expense of DNN copyright protection. 

\begin{figure}[!tbp]
\centering
\subfigure[MROV-V-1.]{
\begin{minipage}[htbp]{0.5\linewidth}
\centering
\includegraphics[width=4cm,height=3cm]{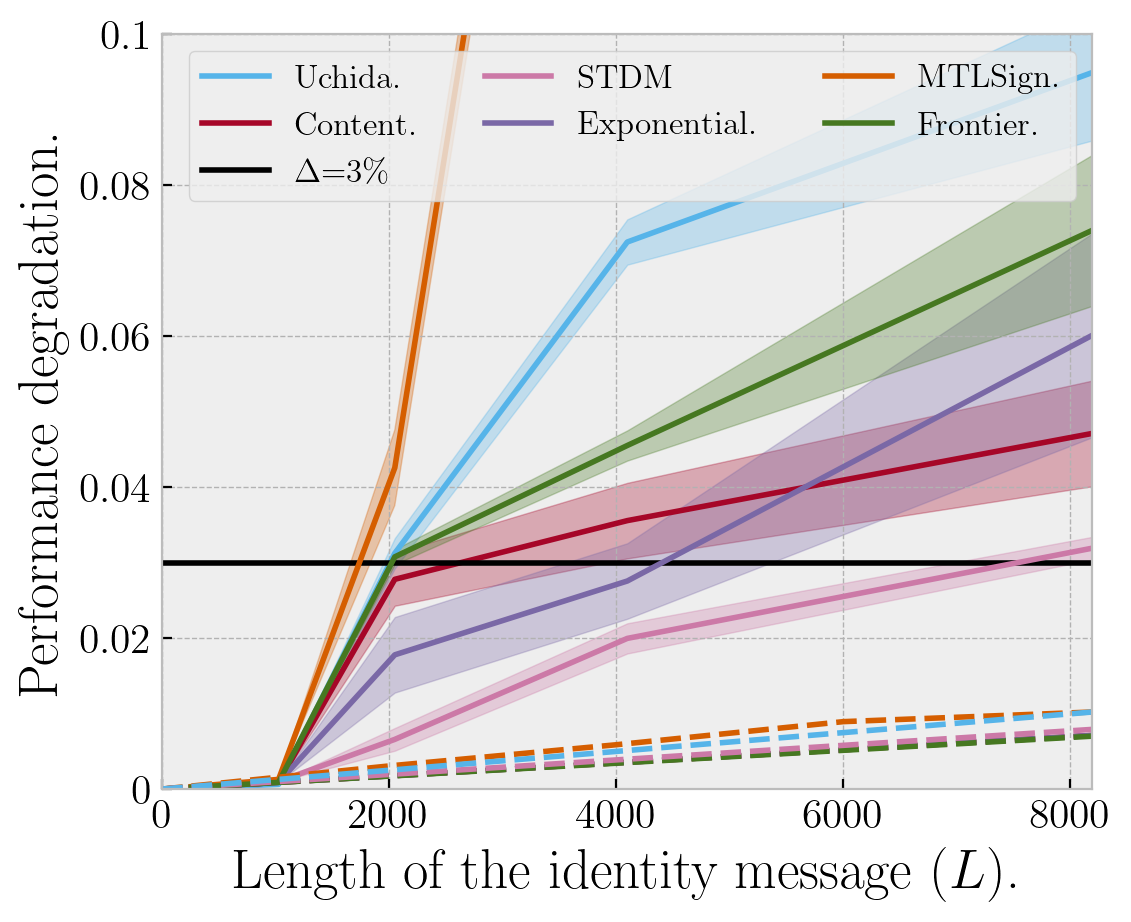}
\end{minipage}
}\subfigure[MROV-V-2.]{
\begin{minipage}[htbp]{0.5\linewidth}
\centering
\includegraphics[width=4cm,height=3cm]{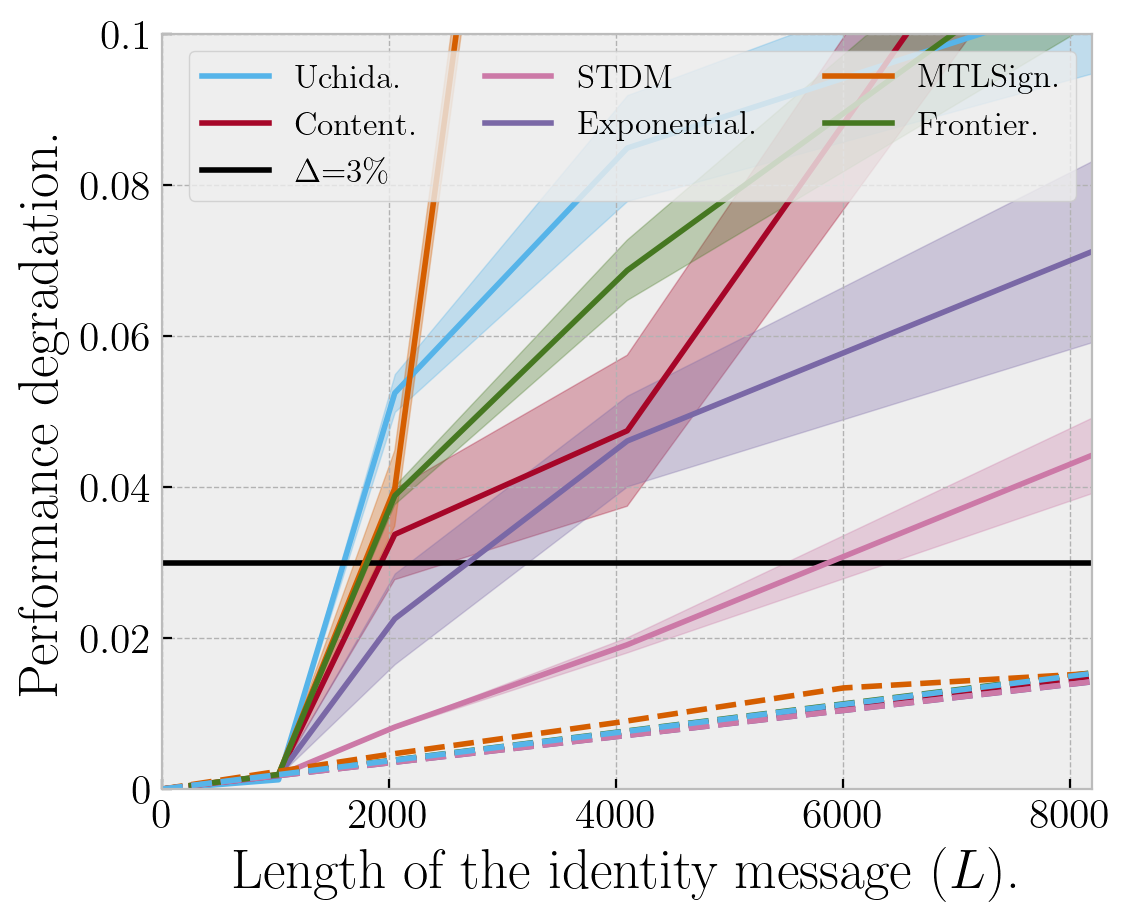}
\end{minipage}
}
\caption{Estimated minimal length of the identity message, $J\!=\!1024$. 
The dashed lines mark the cost in fidelity $F(L)$. 
The solid lines mark  $F(L)\!+\!\min_{\delta}$\{$\delta\!:\!\hat{C}(\delta,L)\!\leq\! J$\}.
}
\label{figure:432}
\end{figure}

\section{Conclusions}
\label{sec:5}
This paper studies the information capacity of DNN watermarks. 
We propose a unified definition and show how it is related to the accuracy, robustness, and the expense of ownership verification. 
Since its closed form is usually intractable, we design an estimation method to find an upper bound of the capacity by adversarial overwriting. 
Finally, we demonstrate that the capacity bottleneck can be broken by reducing the code rate with multiple rounds of ownership verification. 
We incorporate variational approximation on the ownership keys to expand the applicability of this method. 
Extensive experiments show that our scheme efficiently secure the ownership verification with no marginal performance degradation. 

\section*{Acknowledgments}
The authors would like to thank Zhuomeng Zhang and Hangwei Zhang for their assistance in experiments. 
The work described in this paper was supported by the National
Natural Science Foundation of China (62271307) and the Joint Funds of the National Natural Science Foundation of China (U21B2020). 

\bibliography{WM.bib}


\clearpage
\newpage
\appendix
\section{Case Study of Theorem 3 on Classifiers}
We demonstrate the application of Theorem 3 to classifiers with $B$ classes and $D$ training samples.
Assume that the watermarking scheme adopts the vanilla output interpreter, i.e., the prediction of each trigger is mapped into $\log_{2} B$ bits.
So the number of triggers is $\frac{L}{\log_{2} B}$.

Without loss of generality, we adopt the indistinguishable triggers in Exponential Weighting (Namba and Sakuma 2019), i.e., triggers and normal samples follow the same distribution. 
The following discussions adopt the assumption that the adversary has no knowledge of the triggers so it cannot discriminate triggers from normal queries. 
In this case, embedding watermark is equivalent to poisoning $\frac{L}{\log_{2} B}$ out of $D$ training samples, so the fidelity can be estimated by:
\begin{equation}
\label{equation:apfl}
F(L)=\frac{L}{D\cdot\log_{2} B}.
\end{equation}

On the other hand, the indistinguishability condition implies that declining the overall classification accuracy by $\delta$ results in declining the classification accuracy for triggers by $\delta$ (we assume that that adversary has no access to the ownership key for convenience).
For each misclassified trigger, $\frac{B}{B-1}\cdot\frac{\log_{2} B}{2}$ bits of information is lost. 
To prove so, let the binary representation of the original label be $\mathbf{y}$, the mistaken label after attack is $Y$. 
Then the average number of lost bits is:
\begin{equation}
\nonumber
\begin{aligned}
&\sum_{\hat{\mathbf{y}}}\text{Pr}\left(Y=\hat{\mathbf{y}}|\hat{\mathbf{y}}\neq\mathbf{y}\right)\cdot\|\hat{\mathbf{y}}\oplus\mathbf{y} \|_{0}\\
=&\sum_{\hat{\mathbf{y}}}\frac{\text{Pr}\left(Y=\hat{\mathbf{y}},\hat{\mathbf{y}}\neq\mathbf{y}\right)}{\text{Pr}(\hat{\mathbf{y}}\neq\mathbf{y})}\cdot\|\hat{\mathbf{y}}\oplus\mathbf{y} \|_{0}\\
=&\frac{B}{B-1}\cdot\sum_{\hat{\mathbf{y}}\neq\mathbf{y}}\text{Pr}\left(Y=\hat{\mathbf{y}},\hat{\mathbf{y}}\neq\mathbf{y}\right)\cdot\|\hat{\mathbf{y}}\oplus\mathbf{y} \|_{0}\\
=&\frac{B}{B-1}\cdot\left(\sum_{\hat{\mathbf{y}}}\text{Pr}\left(Y=\hat{\mathbf{y}}\right)\cdot\|\hat{\mathbf{y}}\oplus\mathbf{y} \|_{0}-\frac{1}{B}\cdot 0\right)\\
=&\frac{B}{B-1}\cdot\frac{\log_{2}B}{2}
\end{aligned}
\end{equation}
Therefore, the BER under this setting is:
\begin{equation}
\label{equation:apber}
\epsilon_{\delta}=\frac{\frac{L}{\log_{2} B}\cdot\delta\cdot\frac{B}{B-1}\cdot\frac{\log_{2} B}{2}}{L}=\frac{B}{B-1}\cdot\frac{\delta}{2}.
\end{equation}
To solve the equation $C(\delta,L)\!=\!J$, we combine Theorem 2 with Eq.~\eqref{equation:apber} and approximate $H(x)$ with $4x(1-x)$.
As a result, we have:
\begin{equation}
\label{equation:apdelta}
\min_{\delta}\left\{\delta\!:\! C(\delta,L)\leq J \right\}=\frac{B-1}{B}\cdot\left(1-\sqrt{\frac{J}{L}}\right).
\end{equation}

Plugging Eq.~\eqref{equation:apfl} and Eq.~\eqref{equation:apdelta} into Theorem 3, we conclude that $L$ satisfies:
\begin{equation}
\label{equation:apl}
\frac{L}{D\cdot\log_{2} B}+\frac{B-1}{B}\cdot\left(1-\sqrt{\frac{J}{L}}\right)=\Delta.
\end{equation}
The closed form solution to Eq.~\eqref{equation:apl} is:
\begin{equation}
\label{equation:apll}
\tilde{L}\!=\!\left(\sqrt[3]{-\!\frac{q}{2}\!+\!\sqrt{\left(\frac{q}{2}\right)^{2}\!+\!\left(\frac{p}{3}\right)^{3}}}\!+\!\sqrt[3]{-\!\frac{q}{2}\!-\!\sqrt{\left(\frac{q}{2}\right)^{2}\!+\!\left(\frac{p}{3}\right)^{3}}}\right)^{2},
\end{equation}
in which
\begin{equation}
\nonumber
\begin{aligned}
p&=D\cdot\log_{2} B\cdot \left(\frac{B-1}{B}-\Delta\right),\\
q&=-\frac{B-1}{B}\cdot D\cdot\log_{2} B\cdot\sqrt{J}.
\end{aligned}
\end{equation}

The variation of $\tilde{L}$ w.r.t. $J$ and $\Delta$ for $B\!=\!10$, $D\!=\!100000$ is demonstrated in Fig.~\ref{figure:ap1}.
\begin{figure}[!t]
\centering
\includegraphics[width=8.2cm]{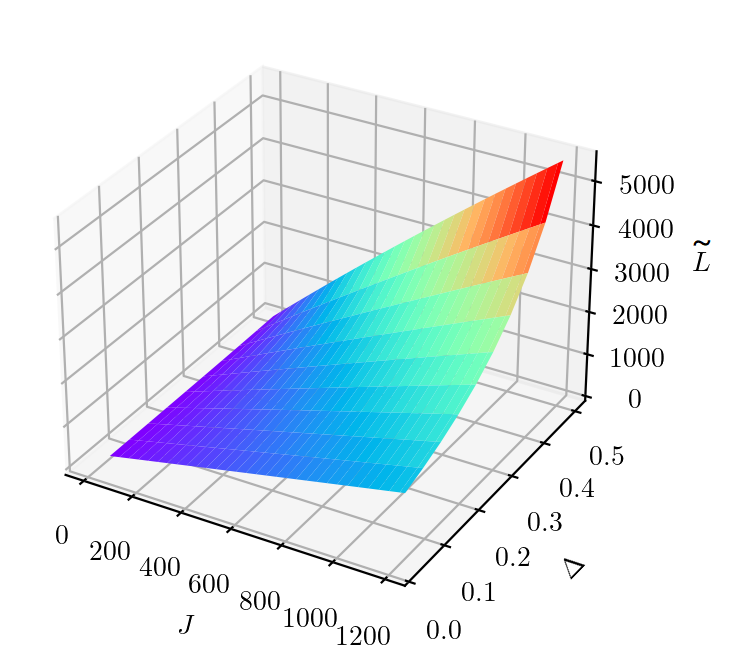}
\caption{The optimal identity message length $\tilde{L}$ w.r.t. the length of the source identity information and the performance degradation range.}
\label{figure:ap1}
\end{figure}

\newpage
\section{Capacity Upper Bound Estimation for Stronger Adversaries}
If the adversary is stronger than the assumption in Section 4, i.e., it possesses more data and can conduct better adversarial overwriting attack, then the capacity upper bound decreases and the expense of copyright protection increases. 

With more data to calibrate the victim model's, the performance degradation corresponding to a certain level of ownership corruption is expected to decline. 
As a result, the capacity upper bound curve would shifted leftward as illustrated in Fig.~\ref{figure:ap2} (a). 

\begin{figure}[htbp]
\centering
\subfigure[]{
\begin{minipage}[htbp]{0.5\linewidth}
\centering
\includegraphics[width=3.6cm,height=3.5cm]{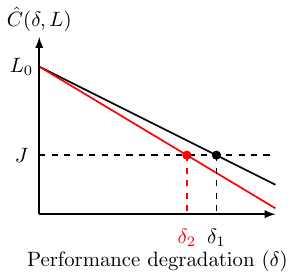}
\end{minipage}
}\subfigure[]{
\begin{minipage}[htbp]{0.5\linewidth}
\centering
\includegraphics[width=4cm,height=3.5cm]{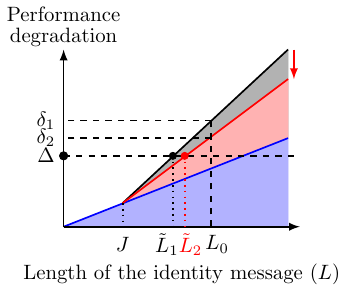}
\end{minipage}
}
\caption{
The results of assuming a stronger adversary. 
(a)~The expected upper bound estimation of the capacity. The red line marks the estimation for the stronger adversary.
(b)~The performance degradation v.s. the length of the identity message as Fig.~\ref{figure:2}, the red line marks the result for the stronger adversary.
}
\label{figure:ap2}
\end{figure}

This phenomenon was observed for all studied schemes as shown in Fig.~\ref{figure:ap4}.

\begin{figure}[!t]
\centering
\subfigure[Uchida.]{
\begin{minipage}[htbp]{0.5\linewidth}
\centering
\includegraphics[width=4cm,height=2cm]{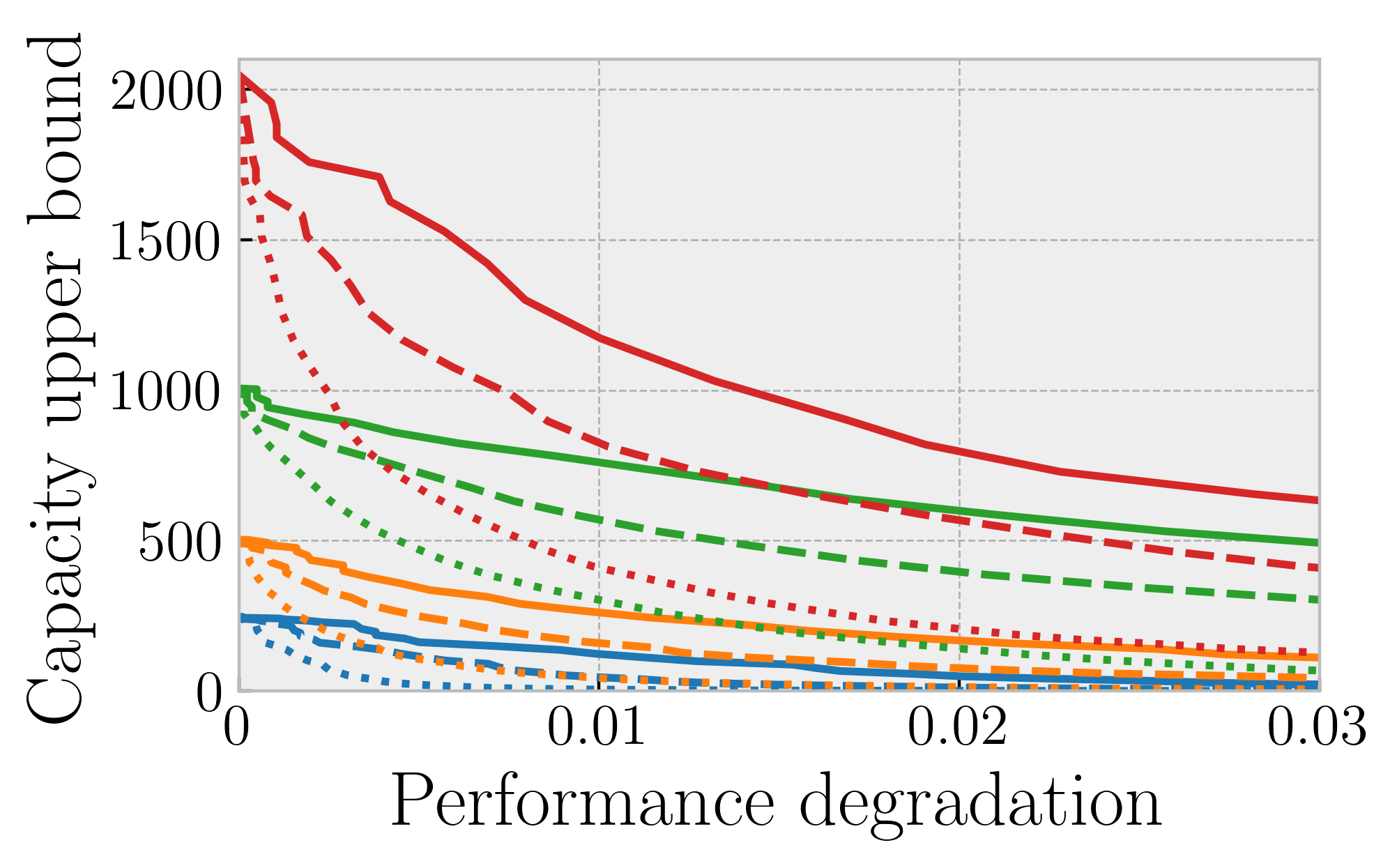}
\end{minipage}
}\subfigure[STDM.]{
\begin{minipage}[htbp]{0.5\linewidth}
\centering
\includegraphics[width=4cm,height=2cm]{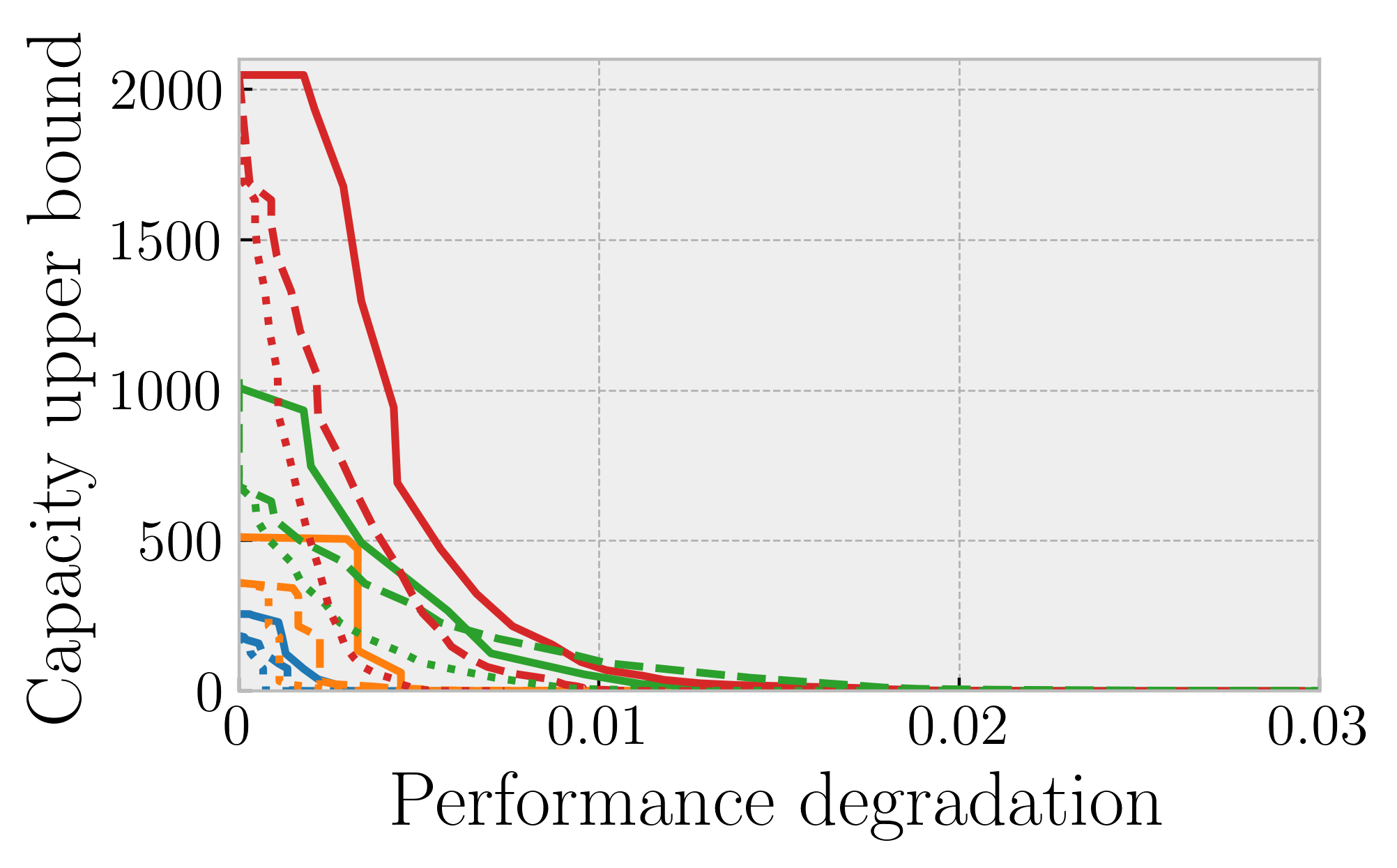}
\end{minipage}
}
\subfigure[MTLSign.]{
\begin{minipage}[htbp]{0.5\linewidth}
\centering
\includegraphics[width=4cm,height=2cm]{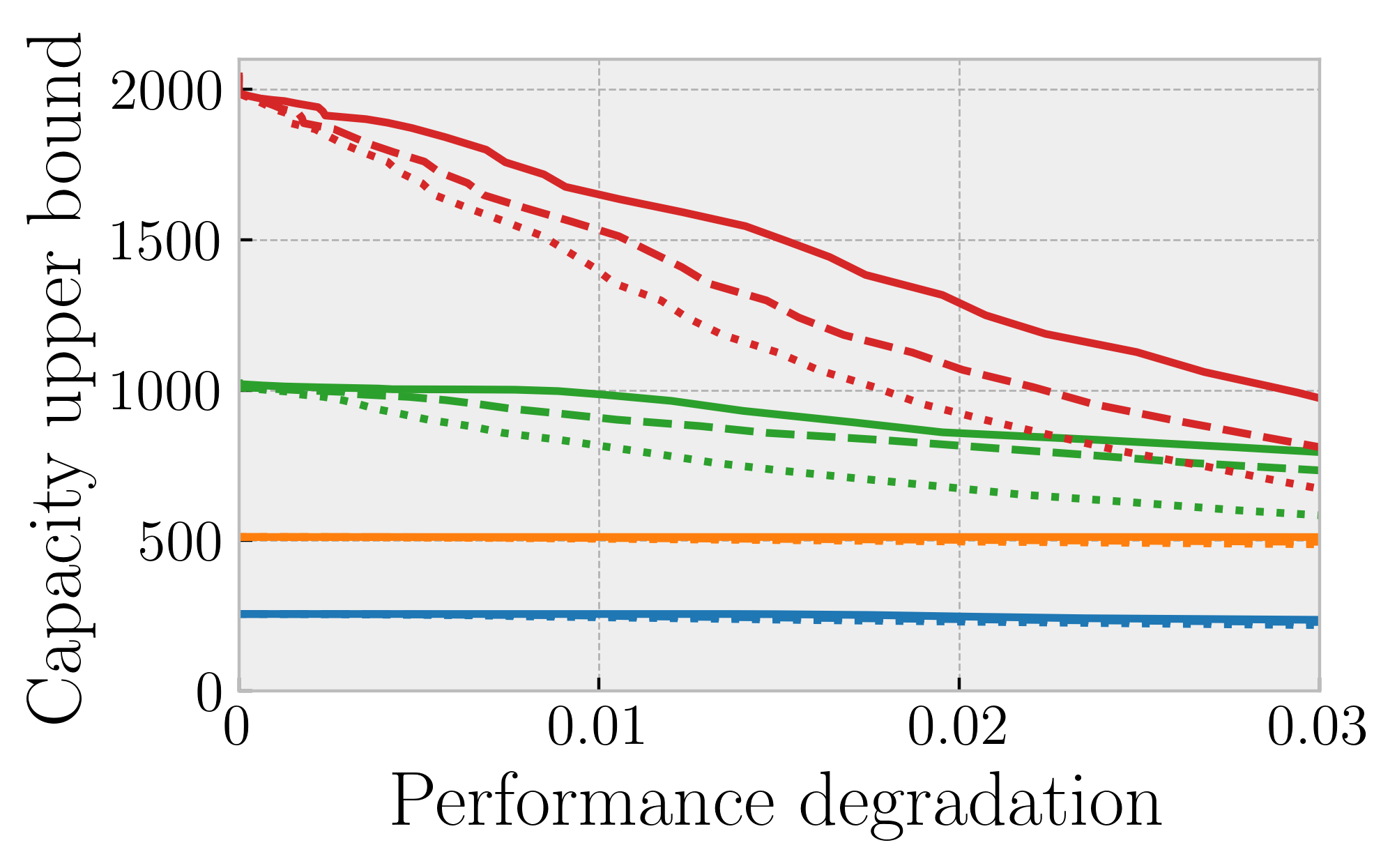}
\end{minipage}
}\subfigure[Content.]{
\begin{minipage}[htbp]{0.5\linewidth}
\centering
\includegraphics[width=4cm,height=2cm]{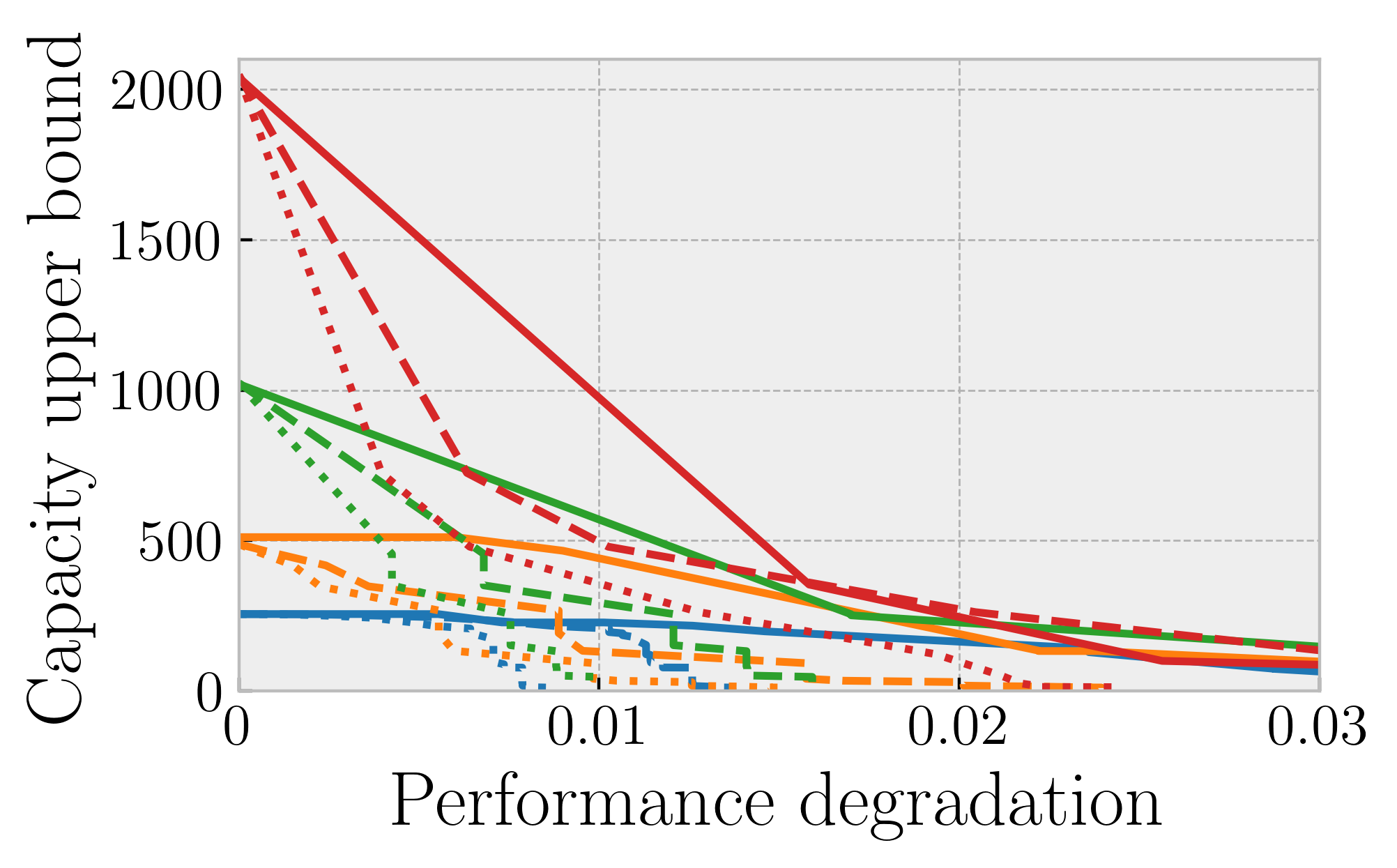}
\end{minipage}
}
\subfigure[Exponential.]{
\begin{minipage}[htbp]{0.5\linewidth}
\centering
\includegraphics[width=4cm,height=2cm]{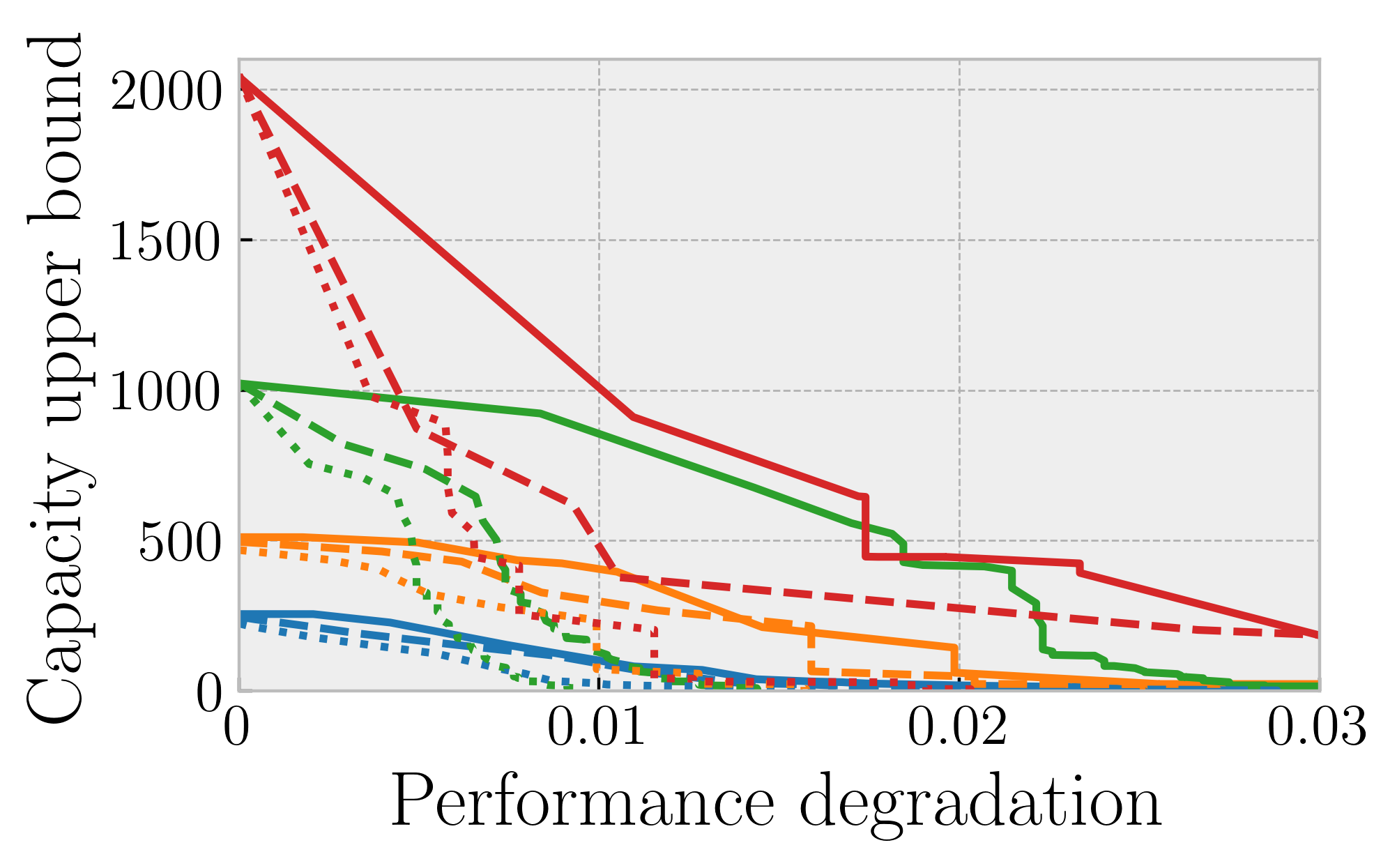}
\end{minipage}
}\subfigure[Frontier.]{
\begin{minipage}[htbp]{0.5\linewidth}
\centering
\includegraphics[width=4cm,height=2cm]{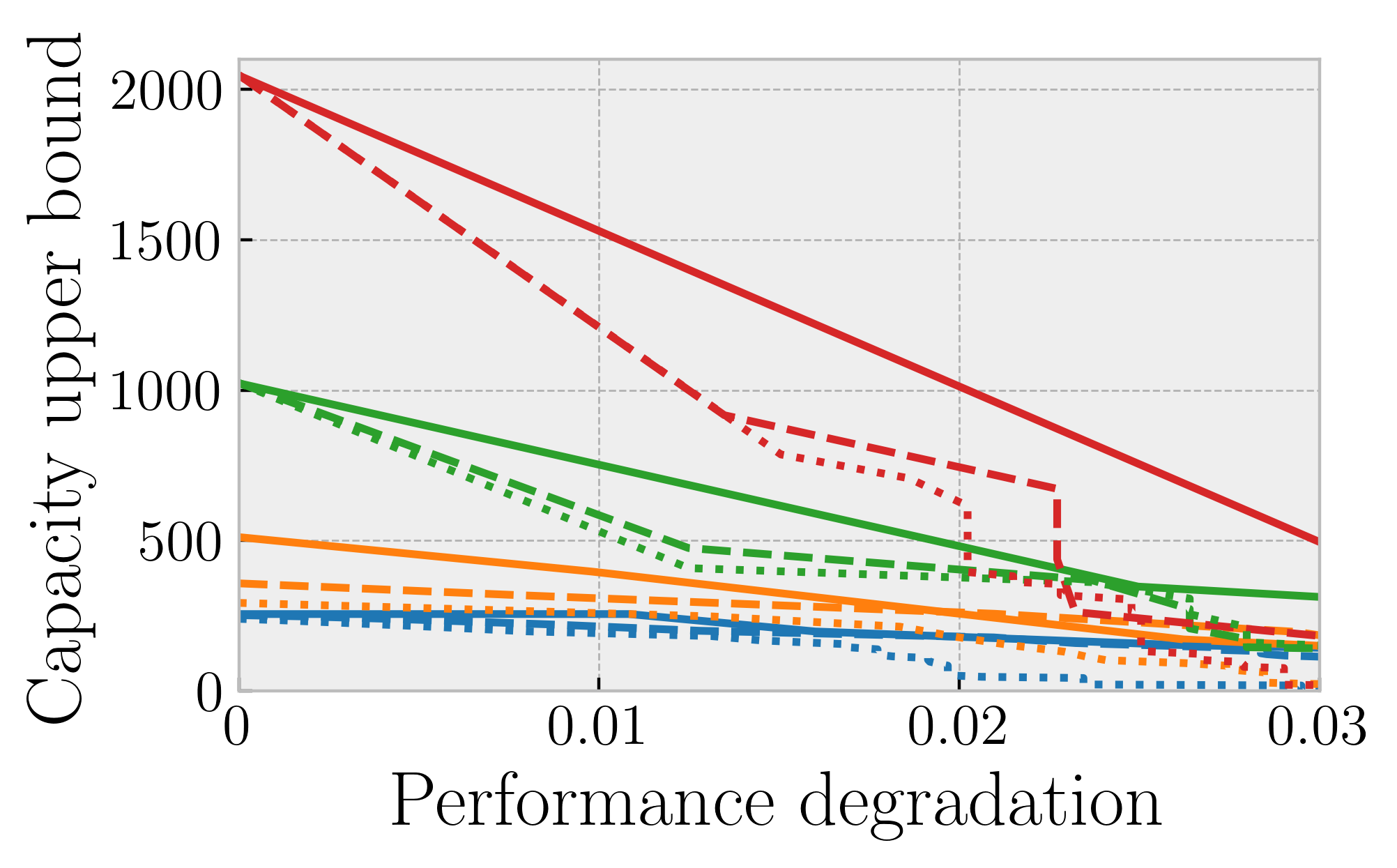}
\end{minipage}
}
\caption{Estimated capacity under adversarial overwriting where the adversary holds 10\% of the training dataset (\protect\tikz[baseline]{\protect\draw[line width=1pt](0,0.5ex)--(0.5,0.5ex);}), 20\% (\protect\tikz[baseline]{\protect\draw[line width=1pt,dashed](0,0.52ex)--(0.5,0.52ex);}), and 30\% (\protect\tikz[baseline]{\protect\draw[line width=1pt,dotted](0,0.5ex)--(0.5,0.5ex);}). 
$L$ is set as {\color{RoyalBlue}256}, {\color{Orange}512}, {\color{LimeGreen}1024}, and {\color{Crimson}2048}.}
\label{figure:ap4}
\end{figure}

Consequently, the value of $\min_{\delta}\left\{\delta:\hat{C}(\delta,L_{0})\leq J \right\}$ is going to decline from $\delta_{1}$ to $\delta_{2}$ for arbitrary $L_{0}$. 
The result of this declining is equivalently depicted in Fig.~\ref{figure:ap2} (b), where it finally results in an increase in the minimal length of the identity message for fixed $J$ and $\Delta$.

The estimated capacity relies on the concrete assumputions on the adversary's knowledge, e.g., the size of the dataset the adversary possesses, whether the adversary's dataset follows the identical distribution or is biased, etc. 
The owner has to determine the configuration of its hypothetical enemy before assessing the capacity of its watermark. 

\end{document}